\documentclass[reqno,11pt]{amsart}
\usepackage{amsmath,amssymb,amsthm,graphicx,a4wide,enumerate,url}
\usepackage[small,bf]{caption} 
\theoremstyle{plain}
\newtheorem{thm}{Theorem}
\newtheorem{lem}[thm]{Lemma}

\theoremstyle{definition}

\theoremstyle{remark}
\newtheorem{rem}{Remark}

\newcommand{\prn}[1]{\left(#1\right)}
\newcommand{\abs}[1]{\left|#1\right|}

\newcommand{\ud}[1]{\,\mathrm{d}#1}
\newcommand{\R}{\mathbb{R}}
\newcommand{\calN}{\mathcal{N}\!{}}

\begin{document}
\parskip.9ex

%===========================================================================
%=================================================================== Titles.
\title[StaRMAP --- Staggered Grid Method for Moment Methods of Radiative Transfer]
{StaRMAP --- A Second Order Staggered Grid Method for Spherical Harmonics Moment Equations of Radiative Transfer}
\author[B. Seibold]{Benjamin Seibold}
\address[Benjamin Seibold]
{Department of Mathematics \\ Temple University \\ \newline
1805 North Broad Street \\ Philadelphia, PA 19122}
\email{seibold@temple.edu}
\urladdr{http://www.math.temple.edu/~{}seibold}
\author[M. Frank]{Martin Frank}
\address[Martin Frank]
{RWTH Aachen University \\ Department of Mathematics \& Center for Computational
Engineering Science \\ Schinkelstrasse 2, D-52062 Aachen, Germany}
\email{frank@mathcces.rwth-aachen.de}
\urladdr{http://www.mathcces.rwth-aachen.de/5people/frank/start}

\subjclass[2000]{65M06; 35L50; 65M12; 35Q20}
%35Q20 Boltzmann equations
%35L50 Initial-boundary value problems for first-order hyperbolic systems
%65M06 Finite difference methods
%65M12 Stability and convergence of numerical methods

\keywords{radiative transfer, method of moments, numerical method, hyperbolic balance law, staggered grid, Matlab}

%===========================================================================
\begin{abstract}
We present a simple method to solve spherical harmonics moment systems, such as the the time-dependent $P_N$ and  $SP_N$ equations, of radiative transfer. The method, which works for arbitrary moment order $N$, makes use of the specific coupling between the moments in the $P_N$ equations. This coupling naturally induces staggered grids in space and time, which in turn give rise to a canonical, second-order accurate finite difference scheme. While the scheme does not possess TVD or realizability limiters, its simplicity allows for a very efficient implementation in \textsc{Matlab}. We present several test cases, some of which demonstrate that the code solves problems with ten million degrees of freedom in space, angle, and time within a few seconds. The code for the numerical scheme, called \textsf{StaRMAP} (\textsf{Sta}ggered grid \textsf{R}adiation \textsf{M}oment \textsf{Ap}proximation), along with files for all presented test cases, can be downloaded so that all results can be reproduced by the reader.
\end{abstract}
%===========================================================================

\maketitle

%===========================================================================
\section{Introduction}
\label{sec:introduction}
%===========================================================================
The purpose of this paper is to present a simple, yet accurate solution method for the $P_N$ equations of radiative transfer, and its efficient implementation in \textsc{Matlab}. The key idea is to make use of the specific coupling of unknowns that is induced by the spherical harmonics being a family of orthogonal polynomials. This leads to a natural staggered grid on which the equations are discretized.

The $P_N$ method (cf.\ \cite{BrunnerHolloway2005}) is one of several ways to discretize the equation of radiative transfer. It is often introduced as an approximate method (method of moments) to reduce the high dimensionality when the kinetic equation of radiative transfer, which is formulated on a six-dimensional domain (one time, two angle, three space), is discretized. Another way of interpreting the $P_N$ equations is to view them as a spectral discretization in the angular variable.

The efficient numerical solution of the $P_N$ equations has become a recent subject of interest \cite{Olson2009,McClarrenHollowayBrunner2008,McClarrenEvansDensmore2008,BrunnerHolloway2005}. The $P_N$ equations have several advantages over other more direct discretizations, such as discrete ordinates, most prominently rotational invariance. The lack of this property leads to the so-called ray effect in discrete ordinates approximations (cf.\ \cite{MorelWareingLowrieParsons2003}).
The key property that the numerical method presented in this work is based upon, is also exclusive to spherical harmonics moment methods, namely a specific coupling structure between the moments. The main drawback of the $P_N$ equations is that they, being a spectral method, can exhibit Gibbs phenomena, i.e., oscillatory behavior that is not present in the solution of the original kinetic equation. In the context of radiative transfer, this can yield negative and therefore unphysical particle densities. In many cases, the Gibbs phenomenon is not a major problem (given the oscillations are small in amplitude), since the $P_N$ equations remain always well-defined. However, sometimes the unphysical particle densities are unacceptable, and several recent works have addressed this fact \cite{HauckMcClarren2010,McClarrenHauck2010}.

The numerical method presented here does not possess limiters in the hyperbolic solver nor does it overcome the Gibbs phenomenon. However, its simplicity allows for fast and very highly resolved computations, so that one can often reduce the spurious oscillations to an acceptable magnitude by choosing the moment order $N$ sufficiently large.

This paper is organized as follows. In Sect.~\ref{sec:PN_slab} we introduce the $P_N$ equations in slab geometry and point out their structure. This is done for didactical reasons, because the notation for the $P_N$ equations in two space dimensions (derived in Sect.~\ref{sec:PN_2d}) becomes quite tedious. The staggered grid method is presented and analyzed in Sect.~\ref{sec:numerical_method}, and in Sect.~\ref{sec:implementation_matlab} the efficient implementation of the numerical scheme in \textsc{Matlab} is outlined. Numerical examples are presented in Sect.~\ref{sec:numerical_results}. We attempt to meet the standards of reproducible research in the computational sciences, laid out by LeVeque \cite{LeVeque2009}. The source code of our package \textsf{StaRMAP} (\textsf{Sta}ggered grid \textsf{R}adiation \textsf{M}oment \textsf{Ap}proximation), along with files to generate all this paper's figures, as well as additional examples, are available to the reader online \cite{StaRMAP}.

%=============================================================================================
\vspace{1.5em}
\section{The Slab Geometry $P_N$ Equations}
\label{sec:PN_slab}
%=============================================================================================
We consider the radiative transfer equation in the form \cite{CaseZweifel1967}
\begin{equation}
\label{eq:RTE}
\begin{split}
\partial_t \psi(t,x,\Omega)
&+ \Omega\cdot\nabla_x\psi(t,x,\Omega)
+ \Sigma_t(t,x)\psi(t,x,\Omega) \\
&= \int_{S^2} \Sigma_s(t,x,\Omega\cdot\Omega') \psi(t,x,\Omega')\ud{\Omega'} + q(t,x,\Omega)\;.
\end{split}
\end{equation}
The quantity $\psi$, which is defined for time $t>0$, space coordinate $x\in\R^3$, and direction $\Omega\in S^2$, is the density of photons that undergo scattering and absorption in a medium. The medium is characterized by the absorption cross section $\Sigma_a$, the scattering kernel $\Sigma_s$ and the total cross section $\Sigma_t = \Sigma_{s0} + \Sigma_a$ ($\Sigma_{s0}$ is defined below). In addition, there is a source $q$. Note that normally in \eqref{eq:RTE}, a factor of $\frac{1}{c}$ appears in front of the time derivative, where $c$ is the speed of light. Here we have set $c = 1$, i.e., we measure time in units of the space scale divided by $c$. Throughout the paper, to return to physical units, time variables have to be multiplied by $c$.

The slab geometry radiative transfer equation is obtained by considering a slab between two infinite parallel plates. Assume for instance that the $z$-axis is perpendicular to the plates. If the setting is invariant under translations perpendicular to, and rotations around, the $z$-axis, then the unknown $\psi$ depends only on the $z$-component of the spatial variable, and one angular variable $\mu$ (cosine of the angle between direction and $z$-axis).

To obtain the $P_N$ equations, we express the angular dependence of the distribution function in terms of a Fourier series,
\begin{equation}
\psi(t,z,\mu) = \sum_{\ell=0}^\infty \psi_\ell(t,z) \tfrac{2\ell+1}{2} P_\ell(\mu)\;,
\end{equation}
where $P_\ell$ are the Legendre polynomials. These form an orthogonal basis of the space of polynomials with respect to the standard scalar product on $[-1,1]$.

One can obtain equations for the Fourier coefficients
\begin{equation}
\psi_\ell = \int_{-1}^1 \psi P_\ell \ud{\mu}
\end{equation}
by testing the radiative transfer equation (\ref{eq:RTE}) with $P_\ell$ and then integrating. Thus we obtain (suppressing the arguments)
\begin{equation*}
\partial_t \psi_\ell+\partial_z\int_{-1}^1 \mu P_\ell \psi \ud{\mu} + \Sigma_{t\ell} \psi_\ell =  q_\ell
\end{equation*}
for $\ell=0,1,\dots$, where
\begin{equation*}
\Sigma_{t\ell} = \Sigma_t - \Sigma_{s\ell} = \Sigma_a + \Sigma_{s0} - \Sigma_{s\ell} \quad\text{and}\quad
\Sigma_{s\ell} = 2\pi \int_{-1}^1 P_\ell(\mu) \Sigma_s(\mu) \ud{\mu}\;.
\end{equation*}

Two properties of the spherical harmonics are crucial for our method. These appear here as properties of the Legendre polynomials.
First, we observe that by the procedure above we have diagonalized the scattering operator on the right hand side (the Legendre polynomials are eigenfunctions of the scattering operator). Second, a general property of orthogonal polynomials is that they satisfy a recursion relation. In particular, the Legendre polynomials $P_\ell$ satisfy
\begin{equation*}
\mu P_\ell(\mu) = \tfrac{\ell}{2\ell+1}P_{\ell-1}(\mu)
+ \tfrac{\ell+1}{2\ell+1}P_{\ell+1}(\mu)\;.
\end{equation*}
Using this fact and truncating the expansion at $\ell=N$ we arrive at the slab-geometry $P_N$ equations
\begin{equation}
\label{eq:PN_slab}
\partial_t \psi_\ell+\partial_z \prn{ \tfrac{\ell+1}{2\ell+1}\psi_{\ell+1}
+ \tfrac{\ell}{2\ell+1}\psi_{\ell-1} } + \Sigma_{t\ell} \psi_\ell = q_\ell\;.
\end{equation}
This system can be written as
\begin{equation*}
\partial_t\vec{u} + M\cdot\partial_z\vec{u} +C\cdot\vec{u} = \vec{q}\;,
\end{equation*}
where
\begin{equation*}
M = \begin{pmatrix}
0 & 1 & & & \\
\frac{1}{3} & 0 & \frac{2}{3} & & \\
& \frac{2}{5} & 0 & \frac{3}{5} & \\
& & \frac{3}{7} & 0 & \ddots \\ & & & \ddots & \ddots
\end{pmatrix}
\quad\text{and}\quad
C = \begin{pmatrix}
\Sigma_{t0} &  &  & & \\
 & \Sigma_{t1} &  & & \\
 &  & \ddots &  & \\
& & &  & \\ & & &  &
\end{pmatrix}\;.
\end{equation*}
The two properties mentioned above lead to
\begin{lem}
\label{lem:slab_even_odd}
The time derivative of $\psi_\ell$ for even (odd) $\ell$ depends only on the spatial derivative of $\psi_k$ for odd (even) $k$, and on the value of $\psi_\ell$ itself.
\end{lem}
Lemma~\ref{lem:slab_even_odd} creates an analogy to the wave equation (with decay), and thus motivates a discretization of the slab geometry $P_N$ equation \eqref{eq:PN_slab} on staggered grids, i.e., all the components with odd $\ell$ are placed in the middle between the components with even $\ell$, and the spatial derivative is approximated by central differences. The numerical scheme presented in Sect.~\ref{sec:numerical_method} generalizes this analogy in the two-dimensional case.

%=============================================================================================
\vspace{1.5em}
\section{The Two-Dimensional $P_N$ Equations}
\label{sec:PN_2d}
%=============================================================================================
In this section, we adopt the notation and the form of the $P_N$ equations as in \cite{BrunnerHolloway2005}. The complex-valued spherical harmonics are defined as
\begin{equation*}
Y_\ell^m(\mu,\phi) = (-1)^m \sqrt{\tfrac{2\ell+1}{4\pi}\tfrac{(\ell-m)!}{(\ell+m)!}}\, e^{im\phi}P_\ell^m(\mu)\;,
\end{equation*}
where $\ell\geq 0$ and $-\ell\leq m\leq \ell$. Here, $P_\ell^m$ are the associated Legendre polynomials. The spherical harmonics form an orthonormal family on the unit sphere. They satisfy a recursion relation of the form
\begin{equation*}
\Omega \overline{Y_\ell^m} = \tfrac{1}{2}
\begin{bmatrix}
-c_{\ell-1}^{m-1}\overline{Y_{\ell-1}^{m-1}} + d_{\ell+1}^{m-1}\overline{Y_{\ell+1}^{m-1}} + e_{\ell-1}^{m+1}\overline{Y_{\ell-1}^{m+1}} - f_{\ell+1}^{m+1}\overline{Y_{\ell+1}^{m+1}} \\
i\prn{c_{\ell-1}^{m-1}\overline{Y_{\ell-1}^{m-1}} - d_{\ell+1}^{m-1}\overline{Y_{\ell+1}^{m-1}} + e_{\ell-1}^{m+1}\overline{Y_{\ell-1}^{m+1}} - f_{\ell+1}^{m+1}\overline{Y_{\ell+1}^{m+1}}} \\
2(a_{\ell-1}^m \overline{Y_{\ell-1}^m} + b_{\ell+1}^m \overline{Y_{\ell+1}^m})
\end{bmatrix}\;,
\end{equation*}
with the coefficients \cite{BrunnerHolloway2005}
\begin{equation*}
\begin{array}{lll}
 a_\ell^m = \sqrt{\frac{(\ell-m+1)(\ell+m+1)}{(2\ell+3)(2\ell+1)}}\;,
&b_\ell^m = \sqrt{\frac{(\ell-m  )(\ell+m  )}{(2\ell+1)(2\ell-1)}}\;,
&c_\ell^m = \sqrt{\frac{(\ell+m+1)(\ell+m+2)}{(2\ell+3)(2\ell+1)}}\;, \\[.8em]
 d_\ell^m = \sqrt{\frac{(\ell-m  )(\ell-m-1)}{(2\ell+1)(2\ell-1)}}\;,
&e_\ell^m = \sqrt{\frac{(\ell-m+1)(\ell-m+2)}{(2\ell+3)(2\ell+1)}}\;,
&f_\ell^m = \sqrt{\frac{(\ell+m  )(\ell+m-1)}{(2\ell+1)(2\ell-1)}}\;.
\end{array}
\end{equation*}
This form already shows a pattern in the coupling of the different moments, that is similar to the slab geometry case.

We multiply (\ref{eq:RTE}) by $\overline{Y_l^m}$, integrate over $\Omega$, and define the expansion coefficients
\begin{equation*}
\psi_\ell^m(t,x) = \int_{S^2} \overline{Y_\ell^m(\Omega)} \psi(t,x,\Omega) \ud{\Omega}\;.
\end{equation*}
As in the slab geometry case, the scattering term becomes diagonal
\begin{equation*}
\int_{S^2} \overline{Y_\ell^m(\Omega)} \int_{S^2} \Sigma_s(\Omega\cdot\Omega') \psi(t,x,\Omega') \ud{\Omega'}\ud{\Omega} =
\Sigma_{s\ell} \psi_\ell^m(t,x),
\end{equation*}
where as before $\Sigma_{s\ell} = 2\pi \int_{-1}^1 P_\ell(\mu) \Sigma_s(\mu) \ud{\mu}$.

Altogether, we obtain the well-known complex-valued $P_N$ equations
\begin{equation}
\label{eq:Pncomplex}
\begin{split}
\partial_t \psi_\ell^m &+\tfrac{1}{2}\partial_x
\prn{-c_{\ell-1}^{m-1}\psi_{\ell-1}^{m-1} + d_{\ell+1}^{m-1}\psi_{\ell+1}^{m-1}
+ e_{\ell-1}^{m+1}\psi_{\ell-1}^{m+1} - f_{\ell+1}^{m+1}\psi_{\ell+1}^{m+1}} \\
&+\tfrac{i}{2}\partial_y
\prn{c_{\ell-1}^{m-1}\psi_{\ell-1}^{m-1} - d_{\ell+1}^{m-1}\psi_{\ell+1}^{m-1}
+ e_{\ell-1}^{m+1}\psi_{\ell-1}^{m+1} - f_{\ell+1}^{m+1}\psi_{\ell+1}^{m+1}} \\
&+ \partial_z
\prn{a_{\ell-1}^m \psi_{\ell-1}^m + b_{\ell+1}^m} + \Sigma_{t\ell} \psi_\ell^m
= q_\ell^m
\end{split}
\end{equation}
for $0\leq \ell < \infty$ and $-\ell\leq m \leq \ell$.

In this work we consider the two-dimensional real-valued $P_N$ equations, which we now derive. There is, however, no conceptual difference to the three-dimensional equations. The reduction is again done via symmetry. This means that we actually solve three-dimensional radiative transfer, but in a geometry that reduces the number of unknowns. For a two-dimensional domain $D$, consider the infinite cylinder $D\times \R\subset\R^3$. We take the angular variable to be aligned with the $z$-direction,
\begin{equation*}
\Omega = (\sqrt{1-\mu^2}\cos\phi,\sqrt{1-\mu^2}\sin\phi,\mu)^T\;.
\end{equation*}
If we assume that all data (coefficients, initial and boundary conditions) are $z$-independent, then the solution $\psi$ is $z$-independent and additionally an even function in $\mu$. Therefore, if $\ell+m$ is odd, the associated Legendre polynomial $P_\ell^m$ is an odd function in $\mu$, and as a consequence the moments for which $\ell+m$ is odd have to vanish. Thus we are left with the (still complex) moments
\begin{equation*}
\psi_0^0,\ \psi_1^{-1},\ \psi_1^1,\ \psi_2^{-2},\ \psi_2^{0},\ \psi_2^{2},\ \dots
\end{equation*}
We thus obtain the following matrix formulation of the $P_N$ equations
\begin{equation*}
\begin{split}
\partial_t &\begin{bmatrix} \psi_0^0\\ \psi_1^{-1}\\ \psi_1^1\\ \psi_2^{-2}\\ \psi_2^{0}\\ \psi_2^{2}\\ \vdots  \end{bmatrix}
+\partial_x \tfrac{1}{2} \begin{bmatrix}
& d_1^{-1} & -f_1^1 & & & & * \\
e_0^0 & & & d_2^{-2} & -f_2^0 & 0 & \\
-c_0^0 & & & 0 & d_2^0 & -f_2^2 & \\
& e_1^{-1} & 0 & & & & * \\
& -c_1^{-1} & e_1^1 & & & & * \\
& 0 & -c_1^1 & & & & * \\
& & & * & * & * & \end{bmatrix}
\cdot
\begin{bmatrix} \psi_0^0\\ \psi_1^{-1}\\ \psi_1^1\\ \psi_2^{-2}\\ \psi_2^{0}\\ \psi_2^{2}\\ \vdots \end{bmatrix} \\
&+ \partial_y \tfrac{i}{2} \begin{bmatrix}
& -d_1^{-1} & -f_1^1 & & & & * \\
e_0^0 & & & -d_2^{-2} & -f_2^0 & 0 & \\
c_0^0 & & & 0 & -d_2^0 & -f_2^2 & \\
& e_1^{-1} & 0 & & & & * \\
& c_1^{-1} & e_1^1 & & & & * \\
& 0 & c_1^1 & & & & * \\
& & & * & * & * & \end{bmatrix}
\cdot
\begin{bmatrix} \psi_0^0\\ \psi_1^{-1}\\ \psi_1^1\\ \psi_2^{-2}\\ \psi_2^{0}\\ \psi_2^{2}\\ \vdots \end{bmatrix} \\
&+ \begin{bmatrix}
\Sigma_{t0} & & & & & & \\
& \Sigma_{t1} & & & & & \\
& & \Sigma_{t1} & & & & \\
& & & \Sigma_{t2} & & & \\
& & & & \Sigma_{t2} & & \\
& & & & & \Sigma_{t2} & \\
& & & & & & \ddots \end{bmatrix}
\cdot
\begin{bmatrix} \psi_0^0\\ \psi_1^{-1}\\ \psi_1^1\\ \psi_2^{-2}\\ \psi_2^{0}\\ \psi_2^{2}\\ \vdots  \end{bmatrix} =
\begin{bmatrix} q_0^0\\ q_1^{-1}\\ q_1^1\\ q_2^{-2}\\ q_2^{0}\\ q_2^{2}\\ \vdots \end{bmatrix}\;.
\end{split}
\end{equation*}
We call the matrix behind the $x$-derivative (including the $\frac12$) $M_x^\text{complex}$, and respectively the matrix behind the $y$-derivative (including the $\frac{i}{2}$) $M_y^\text{complex}$. We denote the matrix containing the $\Sigma_{t\ell}$ by $C$.

The last step is to transform this system to real variables. Note that
\begin{equation*}
\overline{\psi_\ell^m} = (-1)^m \psi_\ell^{-m}\;.
\end{equation*}
Real variables $R_\ell^m$ (for $0\leq m \leq \ell$) and $I_\ell^m$ (for $0 < m \leq \ell$) can be obtained by setting
\begin{equation*}
R_\ell^0 = \psi_\ell^0
\end{equation*}
and for $m\neq 0$
\begin{align*}
R_\ell^m &= \tfrac{(-1)^m }{\sqrt{2}}(\psi_\ell^m + (-1)^m \psi_\ell^{-m})\;, \\
I_\ell^m &= \tfrac{(-1)^mi}{\sqrt{2}}(\psi_\ell^m - (-1)^m \psi_\ell^{-m})\;.
\end{align*}
The factor $(-1)^m$ is chosen so that the coefficients can be compared more easily to well-known moments. For example,
\begin{align*}
R_0^0 &= \tfrac{1}{\sqrt{4\pi}}\int_{4\pi} \psi(\Omega) \ud{\Omega} =  \tfrac{1}{\sqrt{4\pi}}\int_{-1}^1 \int_0^{2\pi} \psi(\mu,\phi) \ud{\phi} \ud{\mu}\;, \\
R_1^1 &= \sqrt{\tfrac{3}{4\pi}}\int_{4\pi} \Omega_x \psi(\Omega) \ud{\Omega} =  \sqrt{\tfrac{3}{4\pi}} \int_{-1}^1 \int_0^{2\pi} \sqrt{1-\mu^2}\cos\phi\,\psi(\mu,\phi) \ud{\phi} \ud{\mu}\;, \\
I_1^1 &= \sqrt{\tfrac{3}{4\pi}}\int_{4\pi} \Omega_y \psi(\Omega) \ud{\Omega} =  \sqrt{\tfrac{3}{4\pi}} \int_{-1}^1 \int_0^{2\pi} \sqrt{1-\mu^2}\sin\phi\,\psi(\mu,\phi) \ud{\phi} \ud{\mu} \;.
\end{align*}
The factor $\sqrt{2}$ is chosen to make the transformation from $\psi_\ell^m$ to $R_\ell^m$, $I_\ell^m$ unitary. If we encode this linear relationship into a matrix $S$, so that
\begin{equation}
\label{eq:transformation}
\vec{u} := \begin{bmatrix} R_0^0\\ R_1^1 \\ I_1^1\\ R_2^{2}\\ I_2^{2}\\ R_2^{0}\\ \vdots  \end{bmatrix} =
S \cdot \begin{bmatrix} \psi_0^0\\ \psi_1^{-1}\\ \psi_1^1\\ \psi_2^{-2}\\ \psi_2^{0}\\ \psi_2^{2}\\ \vdots  \end{bmatrix}\;,
\end{equation}
then we obtain the real-valued system matrices for the $P_N$ equations as
\begin{equation*}
M_x^\text{real} = S M_x^\text{complex} S^{-1}
\quad\text{and}\quad
M_y^\text{real} = S M_y^\text{complex} S^{-1}\;.
\end{equation*}
For example, the $P_3$ matrices are
\begin{equation*}
M_x^\text{real} = \tfrac{1}{2} \begin{bmatrix}
& \sqrt{2}d_1^{-1} & 0 & & & & & & & \\
\sqrt{2}d_1^{-1} & & & d_2^{-2} & 0 & -\sqrt{2}f_2^0 & & & & \\
0 & & & 0 & d_2^{-2} & 0 & & & & \\
& d_2^{-2} & 0 & & & & d_{3}^{-3} & 0 & -f_3^{-1} & 0 \\
& 0 & d_2^{-2} & & & & 0 & d_3^{-3} & 0 & -f_3^{-1} \\
& -\sqrt{2}f_2^{0} & 0 & & & & 0 & 0 & \sqrt{2}d_3^{-1} & 0 \\
& & & d_3^{-3} & 0 & 0 & & & & \\
& & & 0 & d_3^{-3} & 0 & & & & \\
& & & -f_3^{-1} & 0 & \sqrt{2}d_3^{-1} & & & & \\
& & & 0 & -f_3^{-1} & 0 & & & &
\end{bmatrix}
\end{equation*}
and
\begin{equation*}
M_y^\text{real} = \tfrac{1}{2} \begin{bmatrix}
& 0 & \sqrt{2}d_1^{-1} & & & & & & & \\
0 & & & 0 & d_2^{-2} & 0 & & & & \\
\sqrt{2}d_1^{-1} & & & -d_2^{-2} & 0 & -\sqrt{2}f_2^0 & & & & \\
& 0 & -d_2^{-2} & & & & 0 & d_{3}^{-3} & 0& f_3^{-1} \\
& d_2^{-2} & 0 & & & & -d_3^{-3} & 0 & -f_3^{-1} & 0 \\
& 0 & -\sqrt{2}f_2^{0} & & & & 0 & 0 & 0 & \sqrt{2}d_3^{-1} \\
& & & 0 & -d_3^{-3}& 0 & & & & \\
& & & d_3^{-3} & 0 & 0 & & & & \\
& & & 0 & -f_3^{-1} & 0 & & & & \\
& & & f_3^{-1} & 0 & \sqrt{2}d_3^{-1} & & & &
\end{bmatrix}\;.
\end{equation*}
From the previous calculation it can be seen that the matrices $M_x^\text{real}$ and $M_y^\text{real}$ couple the variables $R_\ell^m$ and $I_\ell^m$ in a specific way, as follows.
\begin{lem}
\label{lem:sign_pattern_pn}
The following variables are coupled, provided that they are defined (i.e., for $R_\ell^m$ we must have $\ell\geq 0$, $0\leq m\leq \ell$, for $I_\ell^m$ we must have $\ell\geq 0$, $0<m\leq \ell$):
\begin{itemize}
\item The time-derivative of $R_\ell^m$ depends only on the $x$-derivatives of $R_{\ell\pm 1}^{m\pm 1}$, on the $y$-derivatives of $I_{\ell\pm 1}^{m\pm 1}$, and on $R_\ell^m$ itself.
\item The time-derivative of $I_\ell^m$ depends only on the $x$-derivatives of $I_{\ell\pm 1}^{m\pm 1}$, on the $y$-derivatives of $R_{\ell\pm 1}^{m\pm 1}$, and on $I_\ell^m$ itself.
\end{itemize}
\end{lem}

The scattering matrix $C$ is diagonal and block-wise constant, and therefore invariant under the transformation \eqref{eq:transformation}. Consequently, the real-valued $P_N$ equations read
\begin{equation}
\label{eq:real_valued_pn_equations}
\partial_t\vec{u} + M_x^\text{real}\cdot\partial_x\vec{u}
+ M_y^\text{real}\cdot\partial_y\vec{u} + C\cdot\vec{u} = S\cdot\vec{q}\;,
\end{equation}
where $S\cdot\vec{q}$ contains the real-valued moments of the source $\vec{q}$. We also note that $M_x^\text{real}$ and $M_y^\text{real}$ are both symmetric.

\begin{lem}
\label{lem:L2norm_conservation}
In the absence of $C$, $q$, and boundaries, equation \eqref{eq:real_valued_pn_equations} conserves the global $L^2$ norm of the solution
\begin{equation}
\label{eq:L2_norm}
P[\vec{u}](t) = \prn{\int\int \vec{u}(t,x,y)^T\vec{u}(t,x,y) \ud{x}\ud{y}}^\frac{1}{2}
\end{equation}
over time.
\end{lem}
\begin{proof}
We calculate
\begin{equation*}
\begin{split}
\frac{\text{d}}{\text{d}t} P[\vec{u}]
&= \frac{\text{d}}{\text{d}t} \iint \vec{u}^T\vec{u} \ud{x}\ud{y}
= 2\iint \vec{u}^T \partial_t \vec{u} \ud{x}\ud{y} \\
&= -2\iint ( \vec{u}^TM_x\partial_x \vec{u} + \vec{u}^TM_y\partial_y\vec{u}) \ud{x}\ud{y} \\
&= -\iint \partial_x (\vec{u}^TM_x \vec{u}) + \partial_y (\vec{u}^TM_y\vec{u}) \ud{x}\ud{y}
= 0\;,
\end{split}
\end{equation*}
where the last equality is due to the fact that $M_x$ and $M_y$ are symmetric.
\end{proof}

\begin{rem}
The simplified $P_N$ ($SP_N$) equations derived in \cite{OlbrantLarsenFrankSeibold2013} have the same coupling pattern and can thus be solved with the same numerical scheme, presented in Sect.~\ref{sec:numerical_method}, as the $P_N$ equations.
\end{rem}

\begin{rem}
The question of proper boundary conditions for the $P_N$ equations (and in fact moment models in general) is unsolved. In one space dimension various approaches exist, most prominently Marshak \cite{Marshak1947} or Mark \cite{Mark1944,Mark1945} boundary conditions. For a further discussion and review, we refer the reader to \cite{LarsenPomraning1991}, where asymptotically correct boundary conditions are derived. However, in two and three space dimensions, there is no agreement on the best choice of boundary conditions. Therefore, in this paper we confine ourselves to two types of boundary conditions that are simple to implement: periodic and extrapolation, as described in Sect.~\ref{subsec:boundary_conditions}. It should be pointed out that in many model problems, very little radiation reaches the boundary of the computational domain, thus rendering the choice of boundary conditions irrelevant.
\end{rem}

%===========================================================================
\vspace{1.5em}
\section{Numerical Method}
\label{sec:numerical_method}
%===========================================================================
We now develop a numerical scheme for linear systems of hyperbolic balance laws of the form
\begin{equation}
\label{eq:hyperbolic_balance_law}
\partial_t\vec{u} + M_x\cdot\partial_x\vec{u} + M_y\cdot\partial_y\vec{u}
+C\cdot\vec{u} = \vec{q}\;,
\end{equation}
where the matrices $M_x$, $M_y$, and $C$ possess very specific patterns of their nonzero entries that admit the systematic placement of the components of the solution vector $\vec{u}(x,y,t)$ on staggered grids. Specifically, let the solution of \eqref{eq:hyperbolic_balance_law} have $\calN$ components, i.e., $\vec{u}(x,y,t)\in\mathbb{R}^\calN$, and the source $\vec{q}(x,y,t)\in\mathbb{R}^\calN$ and the matrices $M_x, M_y\in\mathbb{R}^{\calN\times\calN}$ and $C(x,y,t)\in\mathbb{R}^{\calN\times\calN}$ are of appropriate sizes. Moreover, the matrix $C(x,y,t)$ is diagonal, and the matrices $M_x$ and $M_y$ are constant-coefficient, and possess patterns of their nonzero entries, as described in Lemma~\ref{lem:sign_pattern_pn}. Hence, the real-valued $P_N$ equations \eqref{eq:real_valued_pn_equations} are covered, as are the $SP_N$ equations \cite{OlbrantLarsenFrankSeibold2013}.

%---------------------------------------------------------------------------
\subsection{Spatial Approximation on Staggered Grids}
\label{subsec:spatial_approximation}
%---------------------------------------------------------------------------
We consider the partial differential equation \eqref{eq:hyperbolic_balance_law} to hold in the interior of a rectangular computational domain $\Omega = (0,L_x)\times (0,L_y)$, and on each of the two boundary directions (horizontal and vertical), we allow for one of two types of boundary conditions: periodic or extrapolation, as described in more detail below. The domain $\Omega$ is divided into $n_x \times n_y$ rectangular cells of size $\Delta x \times \Delta y$, where $\Delta x = L_x/n_x$ and $\Delta y = L_y/n_y$. The center points of these cells then lie on the grid
\begin{equation}
\label{eq:grid11}
G_{11} = \{((i-\tfrac{1}{2})\Delta x,(j-\tfrac{1}{2})\Delta y) \,|\, i\in\{1,\dots,n_x\},\, j\in\{1,\dots,n_y\}\}\;.
\end{equation}
We always place the first component of the solution vector, i.e.\ the scalar flux $R_0^0$, on this cell-centered grid $G_{11}$. As an example, Figure~\ref{fig:staggered_grid} shows the division of the rectangular domain (light gray) into a $5\times 3$ arrangement of cells. The grid $G_{11}$ is depicted by gray circles. The key principle of the numerical scheme is to place the remaining solution components on grids that are staggered with $G_{11}$.

To reiterate, the condition on the nonzero entry patterns of $M_x$, $M_y$, and $C$, given in Lemma~\ref{lem:sign_pattern_pn}, can be reformulated as follows: the components of $\vec{u}$ can be distributed into four disjoint sets, according to $\{1,2,\dots,\calN\} = I_{11} \,{\cup}\, I_{21} \,{\cup}\, I_{12} \,{\cup}\, I_{22}$, such that the following properties hold:
\begin{equation}
\label{eq:matrix_conditions}
\begin{split}
(M_x)_{i,j} = 0\;\; &\forall\,(i,j)\notin ((I_{11}\times I_{21})
\cup (I_{21}\times I_{11}) \cup (I_{12}\times I_{22}) \cup (I_{22}\times I_{12}))\;, \\
(M_y)_{i,j} = 0\;\; &\forall\,(i,j)\notin ((I_{11}\times I_{12})
\cup (I_{12}\times I_{11}) \cup (I_{21}\times I_{22}) \cup (I_{22}\times I_{21}))\;, \\
C_{i,j} = 0\;\; &\forall\,(i,j)\notin ((I_{11}\times I_{11})
\cup (I_{21}\times I_{21}) \cup (I_{12}\times I_{12}) \cup (I_{22}\times I_{22}))\;.
\end{split}
\end{equation}
With this distribution of the indices of the solution components, we consider four fully staggered grids: $G_{11}$, defined above, and in addition
\begin{equation}
\label{eq:staggered_grids}
\begin{split}
G_{21} &= \{(i\Delta x,(j-\tfrac{1}{2})\Delta y) \,|\, i\in\{p_x,\dots,n_x\},\, j\in\{1,\dots,n_y\}\}\;, \\
G_{12} &= \{((i-\tfrac{1}{2})\Delta x,j\Delta y) \,|\, i\in\{1,\dots,n_x\},\, j\in\{p_y,\dots,n_y\}\}\;, \\
G_{22} &= \{(i\Delta x,j\Delta y) \,|\, i\in\{p_x,\dots,n_x\},\, j\in\{p_y,\dots,n_y\}\}\;,
\end{split}
\end{equation}
where
\begin{equation}
\label{eq:periodicity_flags}
p_x = \begin{cases} 0 & \text{extrapolation b.c.\ in $x$} \\ 1 & \text{periodic b.c.\ in $x$} \end{cases}
\quad\text{and}\quad
p_y = \begin{cases} 0 & \text{extrapolation b.c.\ in $y$} \\ 1 & \text{periodic b.c.\ in $y$} \end{cases}\;.
\end{equation}
In Fig.~\ref{fig:staggered_grid}, the grid $G_{21}$ is depicted by gray top-pointing triangles, the grid $G_{12}$ by gray right-pointing triangles, and the grid $G_{22}$ by gray squares.

\begin{figure}
\centering
\begin{minipage}[b]{.75\textwidth}
\includegraphics[width=\textwidth]{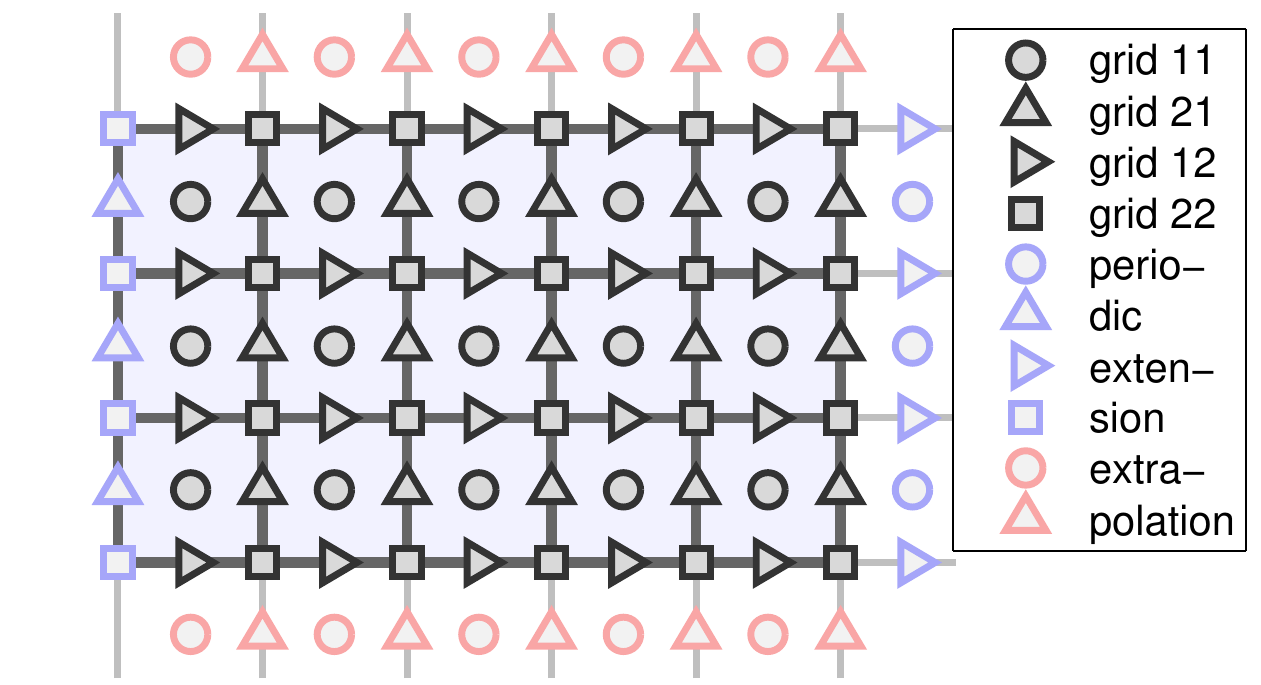}
\end{minipage}
\vspace{-.6em}
\caption{Staggered grid of $5\times 3$ grid cells, with periodic b.c.\ in the $x$-direction, and extrapolation b.c.\ in the $y$-direction. Shown are solution grid points (black boundaries), periodic extension points (light blue), and extrapolation ghost points (light red).}
\label{fig:staggered_grid}
\end{figure}

Having defined these fully staggered grids, the solution components with indices in $I_{k\ell}$ are assigned to the corresponding grid $G_{k\ell}$, where $k,\ell\in\{1,2\}$. On these staggered grids, spatial derivatives of a function $w$ can be approximated by the half-grid central difference approximations
\begin{equation}
\label{eq:central_differences}
\begin{split}
\partial_x w(i\Delta x,j\Delta y) &\approx \tfrac{1}{\Delta x}
\left(w((i+\tfrac{1}{2})\Delta x,j\Delta y)-w((i-\tfrac{1}{2})\Delta x,j\Delta y)\right)
\;\forall\,i,j\in\tfrac{1}{2}\mathbb{Z}\;, \\
\partial_y w(i\Delta x,j\Delta y) &\approx \tfrac{1}{\Delta y}
\left(w(i\Delta x,(j+\tfrac{1}{2})\Delta y)-w(i\Delta x,(j-\tfrac{1}{2})\Delta y)\right)
\;\forall\,i,j\in\tfrac{1}{2}\mathbb{Z}\;,
\end{split}
\end{equation}
and we denote the resulting finite difference operators $D_x$ and $D_y$, respectively. With these finite difference approximations, $x$-derivatives of components on the grid $G_{k\ell}$ are associated with the grid $G_{3-k,\ell}$, and $y$-derivatives of components on the grid $G_{k\ell}$ are placed on the grid $G_{k,3-\ell}$, where $k,\ell\in\{1,2\}$. The nonzero entry patterns \eqref{eq:matrix_conditions} guarantee that the distribution of the indices of $\vec{u}$ into the sets $I_{11}$, $I_{21}$, $I_{12}$, and $I_{22}$ is exactly reproduced by the distribution of the indices of $M_x\cdot D_x\vec{u} + M_y\cdot D_y\vec{u}+C\cdot\vec{u}$. Moreover, the components of the source vector $\vec{q}$ in \eqref{eq:hyperbolic_balance_law} are placed on the same grids as the corresponding components of the solution vector.

%---------------------------------------------------------------------------
\subsection{Boundary Conditions}
\label{subsec:boundary_conditions}
%---------------------------------------------------------------------------
In each of the two coordinate directions, we prescribe one of two types of boundary conditions. In the $x$-axis direction these conditions, and their implementation, are as follows:
\begin{itemize}
\item\textbf{Periodic:} The solution satisfies $\vec{u}(x+L_x,y,t) = \vec{u}(x,y,t)$. On the staggered grids, periodicity is implemented by ``wrapping around'' the grid data, by defining: for all components $k\in G_{21}\cup G_{22}$ set $u_k(0,j\Delta y) := u_k(L_x,j\Delta y)\;\forall\, j\in I_y$, and for all components $k\in G_{11}\cup G_{12}$ set $u_k(L_x+\frac{1}{2}\Delta x,j\Delta y) := u_k(\frac{1}{2}\Delta x,j\Delta y)\;\forall\, j\in I_y$, where $I_y = \{\frac{1}{2}p_y,\frac{1}{2}p_y+\frac{1}{2},\dots,n_y-\frac{1}{2},n_y\}$.
\vspace{.5em}
\item\textbf{Extrapolation:} All solution components with $k\in G_{11}\cup G_{12}$ satisfy $\partial_x u_k = 0$ at the left and the right boundary. On the staggered grids, this condition is implemented by constant extrapolation, i.e., by ``copying'' grid data onto ghost points, by defining: $u_k(-\frac{1}{2}\Delta x,j\Delta y) := u_k(\frac{1}{2}\Delta x,j\Delta y)\;\forall\, j\in I_y$, and $u_k(L_x+\frac{1}{2}\Delta x,j\Delta y) := u_k(L_x-\frac{1}{2}\Delta x,j\Delta y)\;\forall\, j\in I_y$, where $I_y = \{\frac{1}{2}p_y,\frac{1}{2}p_y+\frac{1}{2},\dots,n_y-\frac{1}{2},n_y\}$.
\end{itemize}
The definition and implementation of boundary conditions in the $y$-axis direction works analogously. Figure~\ref{fig:staggered_grid} shows an example with periodic b.c.\ in the $x$-direction (grid data that is wrapped around is depicted in blue) and extrapolation b.c.\ in the $y$-direction (ghost points are depicted in red).

%---------------------------------------------------------------------------
\subsection{Time Stepping}
%---------------------------------------------------------------------------
The time-derivative in \eqref{eq:hyperbolic_balance_law} is resolved by bootstrapping, i.e., data on the grids $G_{11}\cup G_{22}$ is updated alternatingly with data on the grids $G_{21}\cup G_{12}$. This approach is natural, since in the approximate advective part of \eqref{eq:hyperbolic_balance_law}, i.e., $M_x\cdot D_x\vec{u} + M_y\cdot D_y\vec{u}$, the solution components in $I_{11}\cup I_{22}$ (called the ``even components'') depend solely on the components in $I_{21}\cup I_{12}$ (called the ``odd components''), and vice versa. One possibility to perform bootstrapping is by staggering the data in time, i.e., when using a time step $\Delta t$, the even components would be defined at times $n\Delta t$, while the odd components would be defined at times $(n+\frac{1}{2})\Delta t$. While such a space-time-staggered approach yields very elegant update rules, the placement of different solution components at different times causes inconveniences in the implementation of initial conditions, as well as the evaluation of the full solution at a fixed time. Here, we circumvent these problems by defining all solution components on the same time-grid $n\Delta t$, and conducing a time step from $t$ to $t+\Delta t$ via a Strang splitting \cite{Strang1968} of sub-steps, as follows.

For the purpose of the presentation in this section, consider the solution vector be written in block form
\begin{equation*}
\vec{u} = \begin{bmatrix} \vec{u}^\text{e} \\ \vec{u}^\text{o} \end{bmatrix}\;,
\end{equation*}
where $\vec{u}^\text{e}$ is the vector of the even components, and $\vec{u}^\text{o}$ is the vector of the odd components. Moreover, let the same block form apply to the source vector $\vec{q}$ and the matrices $M_x$, $M_y$, and $C$. Then, with the central difference approximations \eqref{eq:central_differences}, equation \eqref{eq:hyperbolic_balance_law} turns into
\begin{equation}
\label{eq:hyperbolic_balance_law_block_form}
\partial_t\begin{bmatrix} \vec{u}^\text{e} \\ \vec{u}^\text{o} \end{bmatrix}
+ \begin{bmatrix} 0 & M_x^\text{eo} \\ M_x^\text{oe} & 0 \end{bmatrix}
\cdot D_x
\begin{bmatrix} \vec{u}^\text{e} \\ \vec{u}^\text{o} \end{bmatrix}
+ \begin{bmatrix} 0 & M_y^\text{eo} \\ M_y^\text{oe} & 0 \end{bmatrix}
\cdot D_y
\begin{bmatrix} \vec{u}^\text{e} \\ \vec{u}^\text{o} \end{bmatrix}
+ \begin{bmatrix} C^\text{e} & 0 \\ 0 & C^\text{o} \end{bmatrix}
\cdot
\begin{bmatrix} \vec{u}^\text{e} \\ \vec{u}^\text{o} \end{bmatrix}
= \begin{bmatrix} \vec{q}^{\,\text{e}} \\ \vec{q}^{\,\text{o}} \end{bmatrix}\;.
\end{equation}
We now define two evolution operators: one that updates the even components only, while ``freezing'' the odd components; and one that updates the odd components only, while ``freezing'' the even components. To advance the even (odd) components from $t$ to $t+\Delta t$, we consider the odd (even) components to be constant on the interval $[t,t+\Delta t]$. The same holds true for the matrix $C$ and the source vector $\vec{q}$. Any quantity that is ``frozen'' during an update from $t$ to $t+\Delta t$ is associated with the intermediate time $t+\tfrac{1}{2}\Delta t$. Specifically, the matrix $C$ and the vector $\vec{q}$ are evaluated at $t+\tfrac{1}{2}\Delta t$.

With these assumptions, the block system \eqref{eq:hyperbolic_balance_law_block_form} decouples into two ODEs
\begin{equation}
\label{eq:hyperbolic_balance_law_decoupled}
\begin{cases}
\partial_t\vec{u}^\text{e}+C^\text{e}\cdot\vec{u}^\text{e} = \vec{r}^{\,\text{e}} \\
\partial_t\vec{u}^\text{o}+C^\text{o}\cdot\vec{u}^\text{o} = \vec{r}^{\,\text{o}}
\end{cases}
\quad\text{with}\quad\
\begin{cases}
\vec{r}^{\,\text{e}} = \vec{q}^{\,\text{e}}-M_x^\text{eo}\cdot D_x\vec{u}^\text{o}
-M_y^\text{eo}\cdot D_y\vec{u}^\text{o} \\
\vec{r}^{\,\text{o}} = \vec{q}^{\,\text{o}}-M_x^\text{oe}\cdot D_x\vec{u}^\text{e}
-M_y^\text{oe}\cdot D_y\vec{u}^\text{e}
\end{cases}\;,
\end{equation}
where $C^\text{e}$, $C^\text{o}$, $\vec{r}^{\,\text{e}}$, and $\vec{r}^{\,\text{o}}$ are time-independent. Moreover, since the matrix $C = \linebreak \text{diag}(c_1,\dots,c_\calN)$ is diagonal, the equations in \eqref{eq:hyperbolic_balance_law_decoupled} decouple into scalar equations of the form
\begin{equation}
\label{eq:scalar_ODE}
\partial_\tau u_k(x,y,\tau) + \bar{c}_k(x,y)u_k(x,y,\tau) = \bar{r}_k(x,y)
\end{equation}
that need to be solved from $\tau = t$ until $\tau = t+\Delta t$. In \eqref{eq:scalar_ODE}, we have $\bar{c}_k(x,y) = c_k(x,y,t+\tfrac{1}{2}\Delta t)$ and $\bar{r}_k(x,y) = r_k(x,y,t+\tfrac{1}{2}\Delta t)$. The exact solution of \eqref{eq:scalar_ODE} is
\begin{align}
\label{eq:sub_step_solution_division}
u_k(x,y,t+\Delta t)
&= \exp(-\bar{c}_k(x,y)\Delta t)u_k(x,y,t)
-\tfrac{1}{\bar{c}_k(x,y)}(1-\exp(-\bar{c}_k(x,y)\Delta t))\bar{r}_k(x,y) \\
\label{eq:sub_step_solution}
&= u_k(x,y,t) + \Delta t\prn{\bar{r}_k(x,y) - \bar{c}_k(x,y) u_k(x,y,t)}
E(-\bar{c}_k(x,y)\Delta t)\;,
\end{align}
where $E(c) = \tfrac{\exp(c)-1}{c}$. Note that, given a robust implementation of this function $E$ (see Sect.~\ref{sec:implementation_matlab}), the representation \eqref{eq:sub_step_solution} has an important advantage over formula \eqref{eq:sub_step_solution_division}, namely it is defined in voids, where $c_k(x,y) = 0$.

Returning to the block form of the solution, we now let $\bar{C}^\text{e}$ and $\bar{C}^\text{o}$ denote the diagonal matrices containing the $\bar{c}_k$ of the even and odd components, respectively. Moreover, the matrices $E(-\bar{C}^\text{e}\Delta t)$ and $E(-\bar{C}^\text{o}\Delta t)$ are the diagonal matrices containing $E(-\bar{c}_k\Delta t)$ of the even and odd components, respectively. Using these notations, formula \eqref{eq:sub_step_solution} gives rise to the even evolution operator
\begin{equation}
\label{eq:solution_operator_even}
S_{t+\Delta t,t}^\text{e}
\begin{bmatrix} \vec{u}^\text{e} \\ \vec{u}^\text{o} \end{bmatrix}
= \begin{bmatrix} \vec{u}^\text{e} \\ \vec{u}^\text{o} \end{bmatrix}
+\Delta t
\begin{bmatrix} E(-\bar{C}^\text{e}\Delta t)
(\vec{r}^{\,\text{e}}-\bar{C}^\text{e}\cdot\vec{u}^\text{e}) \\ 0 \end{bmatrix}
\end{equation}
as the one that advances the even components according to \eqref{eq:sub_step_solution}, and the odd evolution operator
\begin{equation}
\label{eq:solution_operator_odd}
S_{t+\Delta t,t}^\text{o}
\begin{bmatrix} \vec{u}^\text{e} \\ \vec{u}^\text{o} \end{bmatrix}
= \begin{bmatrix} \vec{u}^\text{e} \\ \vec{u}^\text{o} \end{bmatrix}
+\Delta t
\begin{bmatrix} 0 \\ E(-\bar{C}^\text{o}\Delta t)
(\vec{r}^{\,\text{o}}-\bar{C}^\text{o}\cdot\vec{u}^\text{o}) \end{bmatrix}
\end{equation}
as the one that advances the odd components according to \eqref{eq:sub_step_solution}. A full solution step from $t$ to $t+\Delta t$ is then achieved via four sub-half-steps
\begin{equation}
\label{eq:solution_operator_full_step}
S_{t+\Delta t,t} =
S_{t+\frac{1}{2}\Delta t,t}^\text{o}\circ
S_{t+\frac{1}{2}\Delta t,t}^\text{e}\circ
S_{t+\frac{1}{2}\Delta t,t}^\text{e}\circ
S_{t+\frac{1}{2}\Delta t,t}^\text{o}\;.
\end{equation}
Note that in general $S_{t+\frac{1}{2}\Delta t,t}^\text{e} \circ S_{t+\frac{1}{2}\Delta t,t}^\text{e} \neq S_{t+\Delta t,t}^\text{e}$, as one can see from the solution formula \eqref{eq:sub_step_solution}.

%---------------------------------------------------------------------------
\subsection{Accuracy}
\label{subsec:accuracy}
%---------------------------------------------------------------------------
Due to the proper setup of the solution components on staggered grids, all spatial differential operators are approximated in a central fashion, and thus with second order accuracy, i.e., the spatial truncation error is $O(h^2)$, where $h = \max\{\Delta x,\Delta y\}$. Moreover, due to the symmetry in the arrangement of the sub-half-steps in \eqref{eq:solution_operator_full_step}, the splitting error in a single time step is $O(\Delta t^3)$, and henceforth the global temporal truncation error due to the fractional steps is $O(\Delta t^2)$ \cite{Strang1968}. This second order accuracy in time is preserved because all temporal evaluations are done in a symmetric fashion (at the half-step time $t+\tfrac{1}{2}\Delta t$). The overall second order accuracy (i.e., the truncation error is $O(\Delta t^2)+O(h^2)$) is confirmed by the numerical results in Sect.~\ref{subsec:verification}.

%---------------------------------------------------------------------------
\subsection{Stability}
%---------------------------------------------------------------------------
Here, we show the $L^2$ stability at the heart of the approach, namely the combination of staggered spatial grids with the temporal splitting \eqref{eq:solution_operator_full_step}. We conduct the analysis in a simplified case, namely: the $P_1$ system (with advection coefficients set to 1), in one space dimension, with constant coefficients, and without a source (i.e., $\vec{q} = 0$). In this case our solution vector consists of two scalar fields $\vec{u} = \begin{bmatrix} u^\text{e} & u^\text{o} \end{bmatrix}^T$, which satisfy the equations
\begin{equation*}
\begin{cases}
\partial_t u^\text{e} + \partial_x u^\text{o} + c^\text{e}u^\text{e} &= 0 \\
\partial_t u^\text{o} + \partial_x u^\text{e} + c^\text{o}u^\text{o} &= 0\;.
\end{cases}
\end{equation*}
We now conduct a Von Neumann stability analysis of the numerical scheme presented above. To that end, let the numerical solution at time $t$ be represented in terms of its Fourier coefficients
\begin{equation*}
u^\text{e}(x,t) = \sum_{k}a^\text{e}_k(t) e^{ikx}
\quad\text{and}\quad
u^\text{o}(x,t) = \sum_{k}a^\text{o}_k(t) e^{ikx}\;.
\end{equation*}
Note that this particular form of the Fourier representation holds for a periodic domain of length $2\pi$. However, the arguments below transfer to any periodic domain, as well as to an infinite domain (Cauchy problem, for which the sums would turn into integrals).

A step of the numerical scheme \eqref{eq:solution_operator_full_step} in general affects all the Fourier coefficients $a^\text{e}_k$ and $a^\text{o}_k$. However, due to the linearity of the update rule, no mixing occurs between Fourier modes of different frequency, i.e., the coefficients $a^\text{e}_k(t+\Delta t)$ and $a^\text{o}_k(t+\Delta t)$ depend only on the coefficients $a^\text{e}_k(t)$ and $a^\text{o}_k(t)$ for this particular $k$. Consequently, it suffices to investigate the growth factor for basic wave solutions
\begin{equation}
\label{eq:basic_waves}
u^\text{e}(x,t) = a^\text{e}_k(t) e^{ikx}
\quad\text{and}\quad
u^\text{o}(x,t) = a^\text{o}_k(t) e^{ikx}\;.
\end{equation}
Let $u^\text{e}$ be defined on the grid $hj,\,j\in\mathbb{Z}$ and $u^\text{o}$ be defined on the grid $h\,(j+\frac{1}{2}),\,j\in\mathbb{Z}$. Then the (staggered) grid solution values are
\begin{equation*}
u^\text{e}(hj,t) = a^\text{e}_k(t) e^{ikhj}
\quad\text{and}\quad
u^\text{o}(h\,(j+\tfrac{1}{2}),t) = a^\text{o}_k(t) e^{ikh\,(j+\frac{1}{2})}\;,
\end{equation*}
and consequently the staggered grid derivatives are
\begin{align*}
(D_x u^\text{e})(h\,(j+\tfrac{1}{2}),t)
&= a^\text{e}_k(t) \, \frac{e^{ikh\,(j+1)}-e^{ikhj}}{h}
= 2\frac{i}{h}\sin(kh/2) e^{ikh\,(j+\frac{1}{2})} a^\text{e}_k(t)
\intertext{and}
(D_x u^\text{o})(hj,t)
&= a^\text{o}_k(t) \, \frac{e^{ikh\,(j+\frac{1}{2})}-e^{ikh\,(j-\frac{1}{2})}}{h}
= 2\frac{i}{h}\sin(kh/2) e^{ikhj} a^\text{o}_k(t)\;.
\end{align*}
Hence, for the basic waves \eqref{eq:basic_waves} on the staggered grids, the even half-step solution operator \eqref{eq:solution_operator_even} updates the Fourier coefficients by the linear operation
\begin{equation*}
\begin{bmatrix} a^\text{e}_k \\ a^\text{o}_k \end{bmatrix}\!(t+\tfrac{1}{2}\Delta t)
= \underbrace{\begin{bmatrix} d^\text{e} & f^\text{e} \\ 0 & 1 \end{bmatrix}}_{= G^\text{e}_{\frac{1}{2}\Delta t}}
\cdot\begin{bmatrix} a^\text{e}_k \\ a^\text{o}_k \end{bmatrix}\!(t)\;,
\end{equation*}
where $d^\text{e} = \exp(-\frac{1}{2}c^\text{e}\Delta t)$ and $f^\text{e} = -i\frac{\Delta t}{h} E(-\frac{1}{2}c^\text{e}\Delta t) \sin(kh/2)$. Similarly, the odd half-step solution operator \eqref{eq:solution_operator_odd} acts on the Fourier coefficients as
\begin{equation*}
\begin{bmatrix} a^\text{e}_k \\ a^\text{o}_k \end{bmatrix}\!(t+\tfrac{1}{2}\Delta t)
= \underbrace{\begin{bmatrix} 1 & 0 \\ f^\text{o} & d^\text{o} \end{bmatrix}}_{= G^\text{o}_{\frac{1}{2}\Delta t}}
\cdot\begin{bmatrix} a^\text{e}_k \\ a^\text{o}_k \end{bmatrix}\!(t)\;,
\end{equation*}
where $d^\text{o} = \exp(-\frac{1}{2}c^\text{o}\Delta t)$ and $f^\text{o} = -i\frac{\Delta t}{h} E(-\frac{1}{2}c^\text{o}\Delta t) \sin(kh/2)$. Consequently, a full solution time step \eqref{eq:solution_operator_full_step} acts on the Fourier coefficients as
\begin{align}
\nonumber
\begin{bmatrix} a^\text{e}_k \\ a^\text{o}_k \end{bmatrix}\!(t+\Delta t)
&= G^\text{o}_{\frac{1}{2}\Delta t}\cdot G^\text{e}_{\frac{1}{2}\Delta t}
\cdot G^\text{e}_{\frac{1}{2}\Delta t}\cdot G^\text{o}_{\frac{1}{2}\Delta t}
\cdot\begin{bmatrix} a^\text{e}_k \\ a^\text{o}_k \end{bmatrix}\!(t) \\
\label{eq:growth_factor}
&= \underbrace{\begin{bmatrix}
(d^\text{e})^2 + (d^\text{e}+1)f^\text{e}f^\text{o} &
d^\text{o} (d^\text{e}+1)f^\text{e} \\
(d^\text{o}+(d^\text{e})^2)f^\text{o} + (d^\text{e}+1)f^\text{e}(f^\text{o})^2 &
(d^\text{o})^2 + d^\text{o}(d^\text{e}+1)f^\text{e}f^\text{o}
\end{bmatrix}}_{= G}
\cdot\begin{bmatrix} a^\text{e}_k \\ a^\text{o}_k \end{bmatrix}\!(t)\;.
\end{align}
The eigenvalues of the growth factor matrix $G$ are
\begin{equation*}
\lambda_{1,2} = g\pm\sqrt{g^2-(d^\text{e}d^\text{o})^2}\;,
\end{equation*}
where $g = \frac{1}{2}\!\prn{(d^\text{e})^2 + (d^\text{o})^2 + (d^\text{e}+1)(d^\text{o}+1)f^\text{e}f^\text{o}}$ is half the trace of $G$. Since $d^\text{e}$ and $d^\text{o}$ are real, and $f^\text{e}$ and $f^\text{o}$ are purely imaginary, the half-trace $g$ is real.

\begin{thm}
\label{thm:stability}
If $\Delta t<h$, the time stepping \eqref{eq:solution_operator_full_step} yields a stable scheme.
\end{thm}
\begin{proof}
We consider four cases.

\noindent\textbf{Case} $\abs{g} = d^\text{e}d^\text{o}$. \hspace{3em}
In this case both eigenvalues coincide, i.e., $\lambda_1 = \lambda_2 = g$. If a decay is present in the governing equations (i.e., $c^\text{e}>0$ or $c^\text{o}>0$), then $\abs{g} < 1$, and the scheme is absolutely stable. In turn, if no decay is present, one must check that the growth factor matrix $G$ cannot develop a Jordan block. The case $g = 1$ is equivalent to $G$ being the identity matrix, which represents the constant modes that remain unchanged; and the case $g = -1$ cannot occur, because $\frac{\Delta t}{h} < 1$, and thus $\abs{f^\text{e}f^\text{o}} < 1$. Having covered the case of coinciding eigenvalues, in the remaining cases it suffices to show that $\abs{\lambda_{1,2}}\le 1$.

\smallskip\noindent\textbf{Case} $\abs{g} < d^\text{e}d^\text{o}$. \hspace{3em}
One can observe that both eigenvalues $\lambda_{1,2}$ are complex, with real part $g$, and imaginary part $\pm\sqrt{(d^\text{e}d^\text{o})^2-g^2}$. Therefore, their magnitude satisfies
\begin{equation*}
\abs{\lambda_{1,2}}^2 = g^2+\prn{(d^\text{e}d^\text{o})^2-g^2}
= (d^\text{e}d^\text{o})^2 \le 1\;,
\end{equation*}
i.e., the stability condition is satisfied.

\smallskip\noindent\textbf{Case} $g > d^\text{e}d^\text{o} \ge 0$. \hspace{3em}
By directly estimating $g$ we have
\begin{equation*}
\begin{split}
g^2-(d^\text{e}d^\text{o})^2 &= (g - d^\text{e}d^\text{o})(g + d^\text{e}d^\text{o})
\le \tfrac{1}{2}\!\prn{(d^\text{e})^2 + (d^\text{o})^2 - 2d^\text{e}d^\text{o}}
\tfrac{1}{2}\!\prn{(d^\text{e})^2 + (d^\text{o})^2 + 2d^\text{e}d^\text{o}} \\
&= \tfrac{1}{4}(d^\text{e}-d^\text{o})^2 (d^\text{e}+d^\text{o})^2\;,
\end{split}
\end{equation*}
and therefore
\begin{equation*}
\lambda_2 = g+\sqrt{g^2-(d^\text{e}d^\text{o})^2}
\le \tfrac{1}{2}\!\prn{(d^\text{e})^2 + (d^\text{o})^2 + \abs{d^\text{e}-d^\text{o}}(d^\text{e}+d^\text{o})}
= \max\!\left\{(d^\text{e})^2,(d^\text{o})^2\right\} \le 1\;.
\end{equation*}
Moreover, since $g\ge 0$, we have that $\lambda_1 \ge g-g = 0$, hence the stability condition is satisfied.

\smallskip\noindent\textbf{Case} $g < -d^\text{e}d^\text{o} \le 0$. \hspace{3em}
Clearly, $\lambda_2 \le g+\abs{g} = 0$, thus stability is satisfied if $\lambda_1\ge -1$, which is equivalent to the condition $1+(d^\text{e}d^\text{o})^2+2g \ge 0$. The half-trace $g$ depends on the wave number $k$. It achieves its smallest value if $\sin(kh/2)^2 = 1$, thus we can estimate (using the CFL condition $\Delta t/h \le 1$)
\begin{equation*}
\begin{split}
2g &\ge
(d^\text{e})^2 + (d^\text{o})^2
- (1+d^\text{e})E(-\tfrac{1}{2}c^\text{e}\Delta t)\,
(1+d^\text{o})E(-\tfrac{1}{2}c^\text{o}\Delta t) \\
&= (d^\text{e})^2 + (d^\text{o})^2
- \frac{1-(d^\text{e})^2}{\tfrac{1}{2}c^\text{e}\Delta t}\,
\frac{1-(d^\text{o})^2}{\tfrac{1}{2}c^\text{o}\Delta t} \\
&= \exp(-c^\text{e}\Delta t) + \exp(-c^\text{o}\Delta t)
- 4\, \frac{1-\exp(-c^\text{e}\Delta t)}{c^\text{e}\Delta t}\,
\frac{1-\exp(-c^\text{o}\Delta t)}{c^\text{o}\Delta t} \\
&= 2 - \prn{
c^\text{e}\Delta t E(-c^\text{e}\Delta t) + c^\text{o}\Delta t E(-c^\text{o}\Delta t)
+ 4E(-c^\text{e}\Delta t)E(-c^\text{o}\Delta t)
}\;,
\end{split}
\end{equation*}
and therefore (using the notation $t^\text{e} = c^\text{e}\Delta t$ and $t^\text{o} = c^\text{o}\Delta t$) we obtain
\begin{equation*}
1+(d^\text{e}d^\text{o})^2+2g
\ge 4 - 2t^\text{e} E(-t^\text{e}) - 2t^\text{o} E(-t^\text{o})
+ (t^\text{e}t^\text{o}-4) E(-t^\text{e})E(-t^\text{o})
\ge 0\;,
\end{equation*}
where this last inequality can be verified by expanding the power series of the function $E$. This shows that the presented scheme is conditionally stable, under the condition $\Delta t < h$.
\end{proof}

\begin{rem}
As the proof of Thm.~\ref{thm:stability} shows, in the presence of decay terms the scheme is stable under the usual CFL condition $\Delta t \le h$. However, if no decay is present (i.e., one has plain wave propagation), the choice $\Delta t = h$ would lead to the amplification of waves with wave number $k = \frac{\pi}{h}$, i.e., oscillations of a period of two grid cells.
\end{rem}

\begin{rem}
\label{rem:discrete_L2_norm}
It should be stressed that in the absence of decay terms, the preceding proof of stability does \emph{not} imply that the discrete $L^2$ norm of the approximate solution,
\begin{equation}
\label{eq:discrete_L2_norm}
\hat{P}[\vec{u}]
= \prn{\sum_j u^\text{e}(hj)^2 + \sum_j u^\text{o}(h(j+\tfrac{1}{2}))^2}^\frac{1}{2}\;,
\end{equation}
is conserved in time (even though the true solution conserves the $L^2$ norm, see Lemma~\ref{lem:L2norm_conservation}). The reason is that the growth factor matrix $G$, given in \eqref{eq:growth_factor}, is not Hermitian. Since it is diagonalizable (for $\Delta t < h$) with both eigenvalues of magnitude 1, it is power-bounded, i.e., $\|G^k\|_2\le C$, where the constant $C$ is independent of $k$. However, $\|G\|_2>1$, and therefore the discrete $L^2$ norm will generally \eqref{eq:discrete_L2_norm} oscillate in time, with a small amplitude (see the test case in Sect.~\ref{subsec:variations_L2_norm}).
\end{rem}

In analogy with similar types of problems and numerical approaches, one can expect that the stability results of Thm.~\ref{thm:stability} carry over to the general case, albeit with the following modifications:
\begin{itemize}
\item By Duhamel's principle, problems with source terms can be interpreted as a superposition of many homogeneous problems with new initial values. Hence, the stability results carry over to equations with sources.
\item For systems with more than two components that pre-multiply the spatial derivative by a matrix $M$, the maximum admissible time step has to be scaled by the inverse of the largest (in magnitude) eigenvalue of $M$.
\item In two space dimensions, an extra factor of $\frac{1}{2}$ in front of the time step will account for the presence of growth rates in each of the two spatial dimensions.
\item The case of variable coefficients is not covered by von-Neumann analysis. However, as the proof of stability shows, larger decay terms relax the requirements for stability. Hence, it is plausible that the variable coefficient case does not exhibit worse stability properties than a constant coefficient case with the same minimal decay rates.
\end{itemize}
The numerical results in Sect.~\ref{sec:numerical_results} confirm the stability of the method in the general case.

%---------------------------------------------------------------------------
\subsection{Important Advantages of the Methodology}
%---------------------------------------------------------------------------
Besides its aforementioned second order accuracy (in space and time), an important advantage of the presented approach is its structural simplicity and regularity: in each time step and at every grid point, exactly the same operations are performed (albeit with different coefficients), thus giving rise to an efficient vectorization (see Sect.~\ref{sec:implementation_matlab}), and possibly parallelization (see Outlook). Another advantage of the approach is its stable treatment of large decay coefficients (i.e., problems with large absorption and/or scattering): due to the exact solution \eqref{eq:sub_step_solution} of the sub-step ODE \eqref{eq:scalar_ODE}, large values of the decay coefficients $c_k$ do not impose restrictions on the time step. This is in contrast to explicit time-stepping rules (such as Runge-Kutta methods), which would incur severe time step restrictions for such stiff problems.

%---------------------------------------------------------------------------
\subsection{Limitations of the Methodology}
\label{subsec:limitations}
%---------------------------------------------------------------------------
The simplicity of the approach incurs a few fundamental limitations, and the user of \textsf{StaRMAP} should be aware of these limitations when using the code. Most prominently, the use of the $P_N$ closure (or $SP_N$) gives rise to the Gibbs phenomenon (see Sect.~\ref{sec:introduction}), i.e., the exact solution of the moment system may develop oscillations that are absent in the solution of the original radiative transfer equation \eqref{eq:RTE}. As one consequence, the vector of moments may not be realizable, i.e., the moments cannot stem from a non-negative density. This phenomenon is well-known to users of spherical harmonics moment equations, and it is particularly demonstrated in the line source test case in Sect.~\ref{subsec:line_source}.

There is also a Gibbs phenomenon due to having a linear second order numerical scheme, which due to Godunov's limitation theorem \cite{Godunov1959} cannot be monotone, and thus spurious oscillations tend to occur near strong gradients of the solution. These overshoots could be avoided through limiters. However, the addition of limiters would go at the expense of the structural simplicity and efficiency of the method, and they are therefore not included in \textsf{StaRMAP}.

Another limitation of the approach is that its simplicity stems from the regularity in the geometry. Hence, a generalization of the method to irregular or locally refined meshes is not possible without sacrificing some fundamental structural advantages. Similarly, since the method's overall second order accuracy is based on exploiting local symmetries, a generalization to higher orders is impossible without introducing major modifications.

Finally, the temporal splitting could generate spurious oscillations in the case of strong spatial gradients in the material parameters. While the exact treatment of the sub-step ODE \eqref{eq:scalar_ODE} successfully deals with uniformly large decay rates, extreme gradients in the absorption and/or scattering coefficients could trigger $O(1)$ spurious waves, if the time step is not suitably reduced. Again, it is important that the user of \textsf{StaRMAP} be aware of this possibility when attempting to solve problems with large material parameters that in addition exhibit strong gradients or jumps.

%===========================================================================
\vspace{1.5em}
\section{Implementation in Matlab}
\label{sec:implementation_matlab}
%===========================================================================
The \textsf{StaRMAP} project consists of four types of \texttt{m}-files:
\begin{itemize}
\item the solver file \texttt{starmap\_solver.m};
\item files that create moment matrices, such as \texttt{starmap\_closure\_pn.m};
\item example files, such as \texttt{starmap\_ex\_checkerboard.m}; and
\item example creation files, such as \texttt{starmap\_create\_mms.m}.
\end{itemize}
The philosophy is that in all ``usual'' cases, the solver file does not require any modification. Similarly, the files that create moment matrices remain unchanged (the user could add additional types of moment closures via new \texttt{starmap\_closure\_*.m} files, though). What the user provides and modifies are the example files (in which a problem is defined), and/or the example creation files (which create example files). The example files then call the solver file to run the computation. Below, we describe the various components of the implementation in detail.

%---------------------------------------------------------------------------
\subsection{Definition of a Test Case}
%---------------------------------------------------------------------------
Each test case is encoded in an example file (such as \texttt{starmap\_ex\_checkerboard.m}), which defines the problem and then calls the solver file. The communication is achieved via the \textsc{Matlab} struct \texttt{par}, which is created in the example file and then passed to the solver file. This problem parameter can contain
the problem name (\texttt{par.name}),
the moment matrices (\texttt{par.Mx} and \texttt{par.My}),
the moment order index vector (\texttt{par.mom\_order}),
the coordinates of the computational domain (\texttt{par.ax}),
the numbers of grid cells in each coordinate direction (\texttt{par.n}),
the type of boundary conditions (\texttt{par.bc}),
the times of plotting (\texttt{par.t\_plot}),
the final time (\texttt{par.tfinal}),
the CFL number (\texttt{par.cn}),
the list of moments sent to the plotting routine (\texttt{par.mom\_output}),
as well as function handles for
absorption (\texttt{par.sigma\_a}),
scattering (\texttt{par.sigma\_s0} and \texttt{par.sigma\_sm}),
the source (\texttt{par.source}),
initial conditions (\texttt{par.ic}),
and a problem-specific plotting routine (\texttt{par.output}). Any of these parameters that is not provided is set to a default value by the solver file.

A typical \textsf{StaRMAP} example file consists of three parts: (i) the definition of the problem parameters that differ from the default, including function handles; (ii) the creation of the moment matrices (by calling \texttt{starmap\_closure\_pn.m} for $P_N$ and \texttt{starmap\_closure\_spn.m} for $SP_N$) and the call of the solver; and (iii) the definition of the problem-specific functions.

%---------------------------------------------------------------------------
\subsection{Data Structures for Staggered Grids}
%---------------------------------------------------------------------------
At the core of the \textsc{Matlab} implementation of the method developed in Sect.~\ref{sec:numerical_method} is the fact that the differential and evolution operators perform the same type of operation at each grid point. We therefore store each field quantity and each component of the solution vector as a two-dimensional array (which is identical to a matrix in \textsc{Matlab}), in which the first component represents the $x$-index and the second component the $y$-index of the respective grid. Hence, in line with the staggered grids \eqref{eq:grid11} and \eqref{eq:staggered_grids}, every component in the set $I_{11}$ is a matrix of size $n_x\times n_y$, and the components in $I_{21}$, $I_{12}$, and $I_{22}$ are of sizes $(n_x+1-p_x)\times n_y$, $n_x\times (n_y+1-p_y)$, and $(n_x+1-p_x)\times (n_y+1-p_y)$, respectively, where $p_x$ and $p_y$ are the periodicity flags defined in \eqref{eq:periodicity_flags}.

%---------------------------------------------------------------------------
\subsection{Initialization}
%---------------------------------------------------------------------------
A significant part of the \textsf{StaRMAP} code is devoted to the initialization of the data structures. While many of these steps are technical (mostly to ensure that \textsc{Matlab} handles the memory in an efficient fashion), some steps are of conceptual importance:
\begin{itemize}
\item
The staggered grid index sets \texttt{c11}, \texttt{c22}, \texttt{c21}, and \texttt{c12} (representing the respective sets $I_{11}$, $I_{21}$, $I_{12}$, and $I_{22}$ are created automatically from the structure of the moment matrices \texttt{par.Mx} and \texttt{par.My}. This is done by first placing the first solution component in \texttt{c11}, and then continuing to
place components in appropriate sets, whenever a nonzero entry in the matrices \texttt{par.Mx} and \texttt{par.My} requires this, until all components are distributed into the index sets.
\item
The maximal admissible time step is determined by computing the largest (in magnitude) eigenvalue of the moment matrix $M_x$ via the expression
\begin{verbatim}
    abs(eigs(par.Mx,1,'lm'))
\end{verbatim}
Note that it is assumed that the moment system preserves rotational symmetry, and thus the eigenvalues of $M_x$ are identical to the eigenvalues of $n_x M_x + n_y M_y$ for all $n_x^2+n_y^2 = 1$.
\item
The code segment
\begin{verbatim}
    extendx = {[1:par.bc(1),1:n1(1),par.bc(1)*(n1(1)-1)+1],...
        [n2(1)*(1:1-par.bc(1)),1:n2(1)]};
    extendy = {[1:par.bc(2),1:n1(2),par.bc(2)*(n1(2)-1)+1],...
        [n2(2)*(1:1-par.bc(2)),1:n2(2)]};
\end{verbatim}
creates the index vectors (two cell components each) that encode the boundary conditions. The definition of these four index vectors is the only place in the code where the boundary conditions enter.
\item
In the case of isotropic scattering, the vector \texttt{par.mom\_order} is modified to point at the first moment order only, resulting in the three-dimensional cell array of the material parameters to be of length 1 in the third component.
\end{itemize}

%---------------------------------------------------------------------------
\subsection{Spatial Derivatives}
%---------------------------------------------------------------------------
In the \textsf{StaRMAP} code, the moments are assembled in a cell structure \texttt{U}, where the $j^\text{th}$ component is $\texttt{U\{j\}}$. Note that components that are defined on different grids may be of different sizes. The half-grid central difference formulas \eqref{eq:central_differences} applied to the component \texttt{U\{j\}} are therefore computed as
\begin{verbatim}
    dxU{j} = diff(U{j}(extendx{1},:),[],1)/h(1);
    dyU{j} = diff(U{j}(:,extendy{1}),[],2)/h(2);
\end{verbatim}
where the vector \texttt{h} contains the grid spacing $\Delta x$ and $\Delta y$ as components. As described in Sect.~\ref{subsec:spatial_approximation}, if the component \texttt{U\{j\}} is defined on the grid $G_{k\ell}$, then \texttt{dxU\{j\}} is defined on $G_{3-k,\ell}$ and \texttt{dyU\{j\}} is defined on $G_{k,3-\ell}$. Since the \texttt{diff} function reduces one of the array's dimension by 1, the array \texttt{U\{j\}} may need to be extended, before finite differencing is applied. This is encoded in the index vectors \texttt{extendx\{1\}}, \texttt{extendx\{2\}}, \texttt{extendy\{1\}}, and \texttt{extendy\{2\}}, as explained above.

%---------------------------------------------------------------------------
\subsection{Update of Even and Odd Components}
%---------------------------------------------------------------------------
Each time step consists of four half steps. The even solution operator \eqref{eq:solution_operator_even} updates the components on the even grids $G_{11}$ and $G_{22}$, by using spatial derivatives of quantities on the odd grids $G_{21}$ and $G_{21}$; and vice versa for the odd solution operator \eqref{eq:solution_operator_odd}. One step with the even (odd) solution operator consists therefore of two parts: first, spatial differences of the odd (even) components are computed (see above); second, the even (odd) components are updated, via the code segment
\begin{verbatim}
    W = -sumcell([dxU(Ix{j}),dyU(Iy{j})],...
        [par.Mx(j,Ix{j}),par.My(j,Iy{j})]);
    U{j} = U{j}+dt/2*(W+Q{j}-...
        sT{gtx(j),gty(j),par.mom_order(j)}.*U{j}).*...
        ET{gtx(j),gty(j),par.mom_order(j)};
\end{verbatim}
Here the array \texttt{W} contains all $D_x$ and $D_y$ finite differences that influence a certain component, weighted by the negative moment matrices $M_x$ and $M_y$. The call of the function \texttt{sumcell} achieves this task efficiently by only considering derivatives that correspond to actual nonzero entries in the matrices $M_x$ and $M_y$. Since the moment matrices derived in Sect.~\ref{sec:PN_2d} have an $O(1)$ number of nonzero entries per row, this guarantees that the computational effort of the code grows linearly with the number of moments.

The update rule for the $j^\text{th}$ moment is an immediate implementation of the update formula \eqref{eq:sub_step_solution}. The array \texttt{Q\{j\}} represents the $j^\text{th}$ moment of the source term, and the arrays \texttt{sT\{gtx(j),gty(j),par.mom\_order(j)\}} and \texttt{ET\{gtx(j),gty(j),par.mom\_order(j)\}} encode the decay quantity $c_j(x,y)$ and $E(-c_j(x,y)\frac{\Delta t}{2})$, respectively. Here, \texttt{sT} stands for $\Sigma_t$, and \texttt{ET} denotes the function $E$ being applied to $\Sigma_t$. There is a small modification to this update rule for $j=1$: since the zeroth moment is unaffected by scattering, the update is conducted with the fields \texttt{sA} (which stands for $\Sigma_a$) and \texttt{EA} (which is the function $E$ applied to $\Sigma_a$).

Structurally, the decay terms \texttt{sT} and \texttt{ET} are arranged in three dimensional cell arrays. The first two components encode the grid index in each coordinate direction (1 for odd, and 2 for even). The third component encodes the moment order index $\ell$, given in the index vector \texttt{par.mom\_order}, which encodes a mapping from the system component index to the moment order $\ell$. Since different moments of the same moment order $\ell$ possess the same decay parameters, this structure guarantees that the number of decay quantities grows only linearly in the moment \emph{order} $N$.

%---------------------------------------------------------------------------
\subsection{expm1div}
%---------------------------------------------------------------------------
The function $E(c) = \frac{\exp(c)-1}{c}$, used in the solution formula \eqref{eq:sub_step_solution}, requires a careful implementation. For values of $c$ away from zero, this formula can be implemented as given. However, for $\abs{c}\ll 1$ the difference and division may lead to severely amplified round-off errors, which to leading order equal $\frac{\delta}{c}$, where $\delta$ is the machine accuracy. In the \textsf{StaRMAP} code (function \texttt{expm1div}), the exponential formula is therefore replaced by the Taylor series approximation $E(c) = 1-\frac{1}{2}c+\frac{1}{6}c^2$ for $\abs{c}\le 2\times 10^{-4}$. While this series approximation does not yield any amplified round-off errors, it possesses an approximation error that to leading order equals $\frac{1}{24}c^3$. For double-precision arithmetics, the total errors incurred with these two evaluation formulas match for $\abs{c}\approx 2\times 10^{-4}$, at a value of less than $10^{-12}$.

%---------------------------------------------------------------------------
\subsection{Evaluation of Material Parameters and Source}
%---------------------------------------------------------------------------
In order for the aforementioned update step to have the decay quantities \texttt{sA} and \texttt{sT}, as well as the source \texttt{Q}, available, these quantities must be evaluated before the update. While the \textsf{StaRMAP} code can treat rather general problem setups (anisotropic and/or time-dependent parameters), an important focus lies on its efficiency in special situations, and a significant part of the code is devoted to properly addressing this demand, as follows:
\begin{itemize}
\item
Scattering, absorption, and the source term are each evaluated once in the beginning of each time step, at time $t+\frac{1}{2}\Delta t$. However, if any of these quantities is actually time-independent, then its evaluation occurs in the first time step only.
\item
The scattering is divided into its isotropic component (encoded in \texttt{par.sigma\_s0}) and the anisotropic component (encoded in \texttt{par.sigma\_sm}), where $m$ stands for the moment order (see above). For a problem with isotropic scattering, one simply does not provide the function \texttt{par.sigma\_sm}. In turn, if this function is provided, then it is evaluated for all moment orders $1,\dots,N$. Note that the function \texttt{par.sigma\_sm} is never evaluated for $m=0$; scattering never influences the zeroth moment.
\item
The source function \texttt{par.source} allows for a component for each moment. However, since in many applications sources are acting on only a few moments, one can encode the indices of these moments in the vector \texttt{par.source\_ind}. Specifically, an isotropic source would apply only to the zeroth moment, in which case \texttt{par.source\_ind = 1}.
\end{itemize}

%---------------------------------------------------------------------------
\subsection{Time Stepping}
%---------------------------------------------------------------------------
Each time step consists of three parts: (i) the evaluation of material parameters and sources, if applicable; (ii) the application of the four half-step solution operators, given in \eqref{eq:solution_operator_full_step}; and (iii) the plotting of the solution, if applicable.

In the application of the half-step operators, the even solution operator $S_{t+\frac{1}{2}\Delta t,t}^\text{e}$ is applied twice in succession. While one cannot replace these two half-steps by one full step (if one did, one would lose the second order in time), the finite differences need to be evaluated only once. This is done in the \textsf{StaRMAP} code, thus reducing the number of operations. Note that in the special case of a fully time-independent problem, one could apply the same trick to the odd solution operator (combined over two successive time steps). Since this case is quite rare, this trick is not implemented.

The solution (or, in most cases, the zeroth moment of the solution) is plotted at all times given in the vector \texttt{par.t\_plot}. At the end of each time step, for all times in \texttt{par.t\_plot} that lie in the interval $(t,t+\Delta t]$, the solution is computed via linear interpolation in time $u(t_\text{plot}) = (1-\lambda)u(t)+\lambda u(t+\Delta t)$, where $\lambda = \frac{t_\text{plot}-t}{\Delta t}$. This preserves the second order accuracy of the solution. With this approach, the computation is completely unaffected by the output of the solution.

%=============================================================================================
\vspace{1.5em}
\section{Numerical Results}
\label{sec:numerical_results}
%=============================================================================================

%---------------------------------------------------------------------------
\subsection{Verification via Manufactured Solutions}
\label{subsec:verification}
%---------------------------------------------------------------------------
We use the method of manufactured solutions \cite{SalariKnupp2000} to verify our implementation and to validate the accuracy predictions about the numerical approach made in Sect.~\ref{subsec:accuracy}. To that end, a routine is implemented that uses \textsc{Matlab}'s symbolic toolbox to automatically compute the source vector $\vec{q}(x,y,t)$ that generates a prescribed solution under prescribed material coefficients. We select a smooth, but spatially and temporally varying, solution and smooth, but non-constant, coefficients. Moreover, the solution and the coefficients (and consequently the sources) are periodic on the rectangular computational domain.

As an example, we consider the $P_3$ model for $t\in [0,0.5]$ on the domain $(x,y)\in [0,1]\times[0,1]$. We use a time- and space-dependent absorption coefficient $\Sigma_a(t,x) = t\cos(2\pi x_2)$, and anisotropic Henyey-Greenstein scattering \cite{HenyeyGreenstein1941} with $\Sigma_{s\ell} = 0.9^\ell$. The manufactured solution is
\begin{equation*}
R_0^0(t,x) = e^{-t}\sin^2(2\pi x_1)
\end{equation*}
with all other moments being zero.

The problem is computed for 4 different grid resolutions, and at the final time the difference between the approximate and the true solution for the 10 moments is evaluated. Figure~\ref{fig:convergence} displays the spatial $L^1$, $L^2$, and $L^\infty$ errors as functions of the grid resolution. The four panels show the errors in the scalar flux $R_0^0$ (top left), the first order moments $R_1^1$, $I_1^1$ (top right), the second order moments $R_2^2$, $I_2^2$, $R_2^0$ (bottom left), and the third order moments $R_3^3$, $I_3^3$, $R_3^1$, $I_3^1$ (bottom right). All moments exhibit the expected second order convergence.

\begin{figure}
\centering
\begin{minipage}[b]{.90\textwidth}
\includegraphics[width=\textwidth]{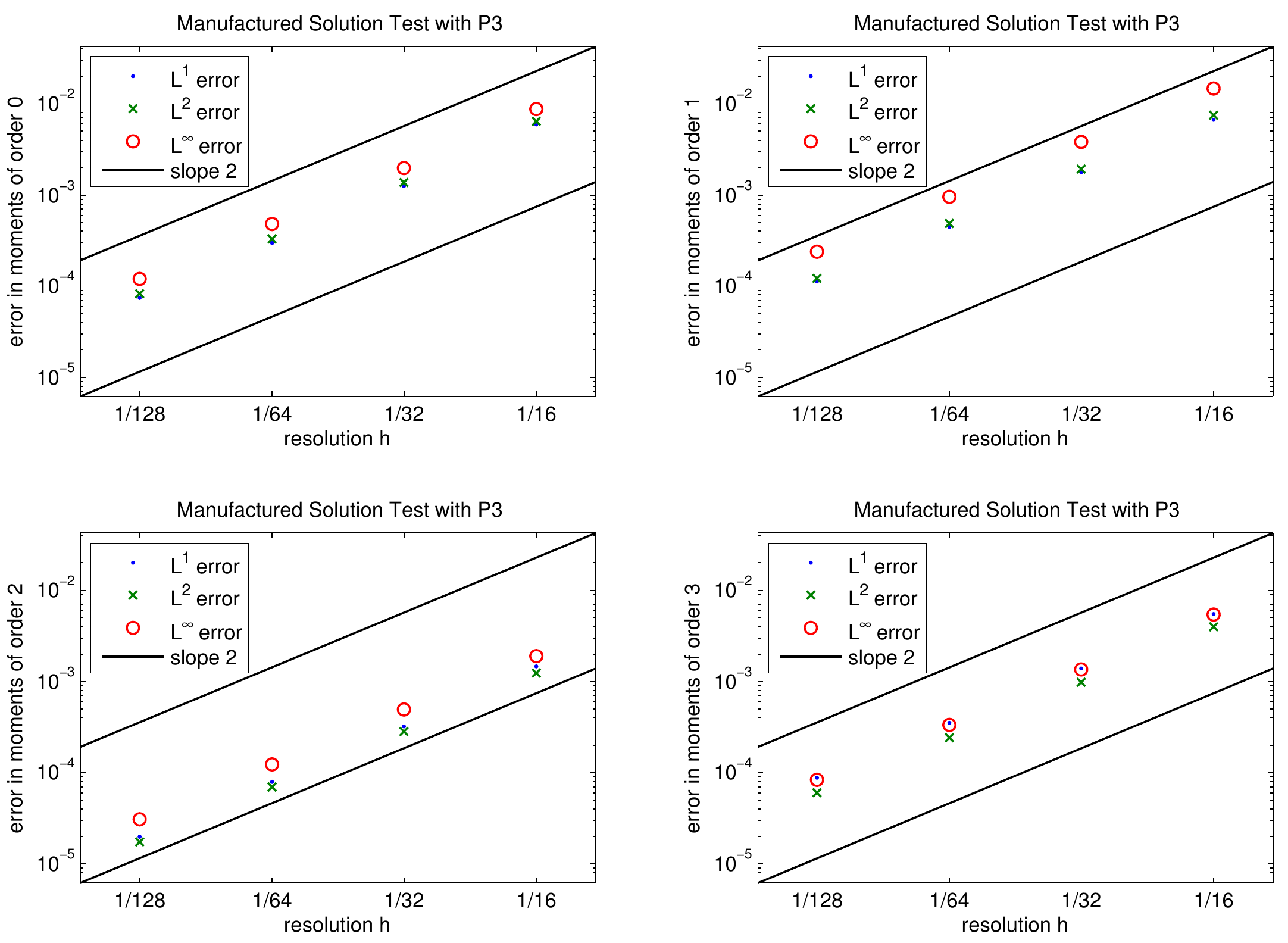}
\end{minipage}
\vspace{-.6em}
\caption{Convergence study with manufactured solution. Errors (in $L^1$, $L^2$, and $L^\infty$) as functions of the grid resolution $h$. Top left: zeroth moment; top right: first order moments; bottom left: second order moments; bottom right: third order moments.}
\label{fig:convergence}
\end{figure}

In the \textsf{StaRMAP} project, the ``manufacturing'' of a solution, i.e., the computation of a source term that generates a prescribed solution, in done via the file \texttt{starmap\_create\_mms.m}. This m-file generates the actual example file \texttt{starmap\_ex\_mms\_auto.m}, which then can be executed to produce the plot shown in Fig.~\ref{fig:convergence}.

%---------------------------------------------------------------------------
\subsection{Variations in the Discrete $L^2$ Norm}
\label{subsec:variations_L2_norm}
%---------------------------------------------------------------------------
In the absence of absorption, scattering, sources, and boundary conditions, the true solution of the $P_N$ equations conserves the $L^2$ norm \eqref{eq:L2_norm} exactly (see Lemma~\ref{lem:L2norm_conservation}). In contrast, as pointed out in Remark~\ref{rem:discrete_L2_norm}, the numerical scheme does not conserve the discrete $L^2$ norm \eqref{eq:discrete_L2_norm} of the approximate solution exactly.

In this test we study the magnitude of temporal variations of the discrete $L^2$ norm of the numerical solution. We consider the $P_5$ equations in the domain $[-1,1]\times [-1,1]$ with periodic boundary conditions, and the following initial conditions. The scalar flux $R_0^0$ is a Gaussian \eqref{eq:Gaussian} with $\sigma = 10^{-2}$, and all other moments are zero initially. The waves propagate in a void without any sources. The problem is computed on three different resolutions: a $50\times 50$ grid, a $100\times 100$ grid, and a $200\times 200$ grid.

\begin{figure}
\centering
\begin{minipage}[b]{.95\textwidth}
\includegraphics[width=\textwidth]{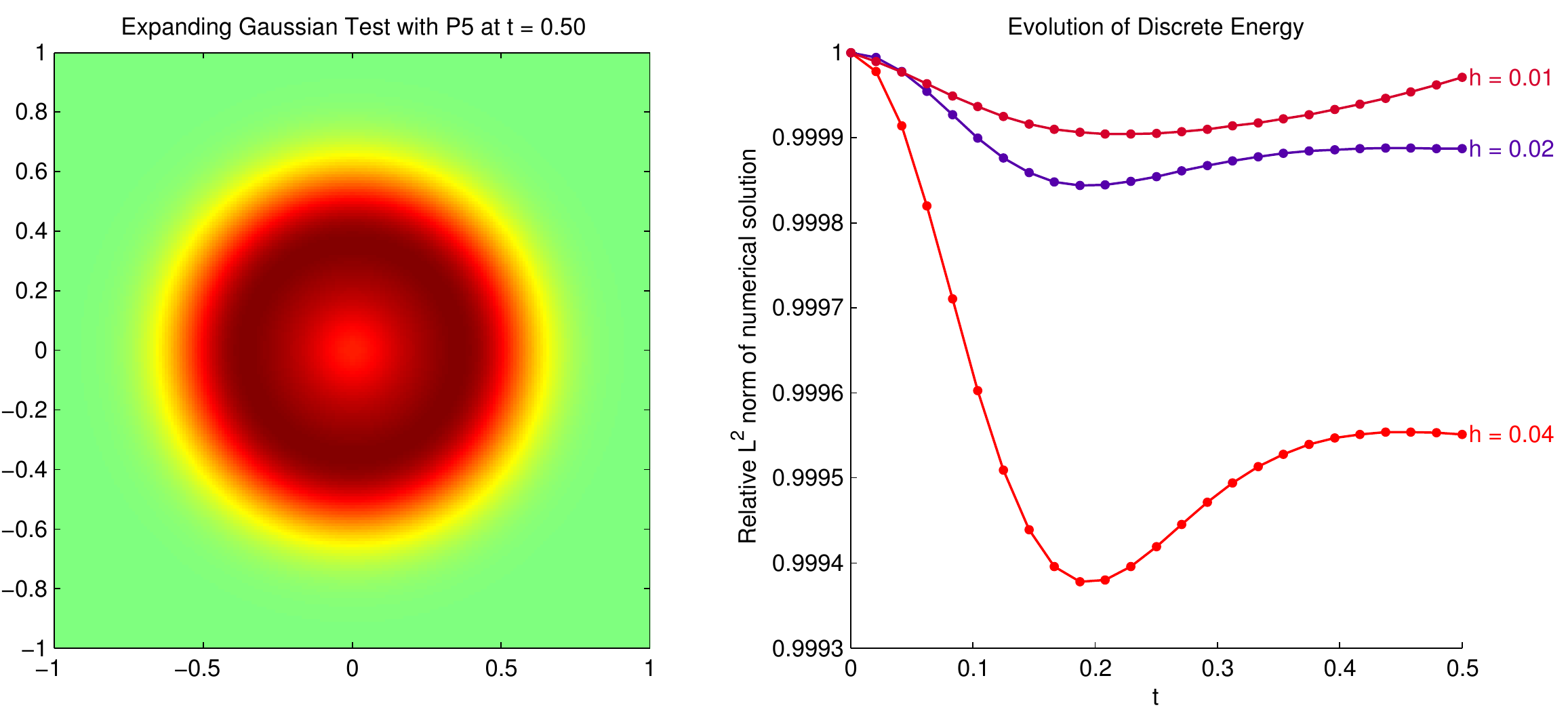}
\end{minipage}
\vspace{-.6em}
\caption{Investigation of the temporal variations in the discrete $L^2$ norm of the numerical solution. An initial Gaussian spreading in a void (left panel) is computed with the $P_5$ moment system. For three different mesh resolutions, the temporal evolution of the $L^2$ norm is recorded (right panel). One can observe that (a) the variations in the discrete $L^2$ norm are small (less than 0.02\% on a $100\times 100$ grid); and (b) the variations diminish as the grid is refined.}
\label{fig:l2norm}
\end{figure}

Figure~\ref{fig:l2norm} displays the results of this test. The scalar flux at the final time ($t = 0.5$) is shown in the left panel, and the temporal evolution of the discrete $L^2$ norm is shown in the right panel (normalized to the value 1 at $t = 0$). The results confirm that the discrete $L^2$ norm is not conserved exactly. However, they also demonstrate that its variations are very small: less than 0.02\% on a $100\times 100$ grid. Moreover, the amount of variation decays to zero as the grid is refined. In the \textsf{StaRMAP} project, this test is implemented in the example file \texttt{starmap\_ex\_l2norm.m}.

%---------------------------------------------------------------------------
\subsection{Checkerboard}
%---------------------------------------------------------------------------
We consider the checkerboard problem described in \cite{BrunnerHolloway2005}. The computational domain is a square of size $[0,7]\times [0,7]$ where the majority of the region is purely scattering. In the middle of the lattice system, in the square $[3,4]\times [3,4]$, an isotropic source $q=1$ is continuously generating particles. Additionally, there are eleven small squares of side length 1 of purely absorbing spots in which $\Sigma_a=10=\Sigma_t$. In the rest of the domain, $\Sigma_a=0$, $\Sigma_t=1$. Figure~\ref{fig:checkerboard_geometry} illustrates the problem setting more precisely.

\begin{figure}
\centering
\begin{minipage}[b]{.40\textwidth}
\includegraphics[width=\textwidth]{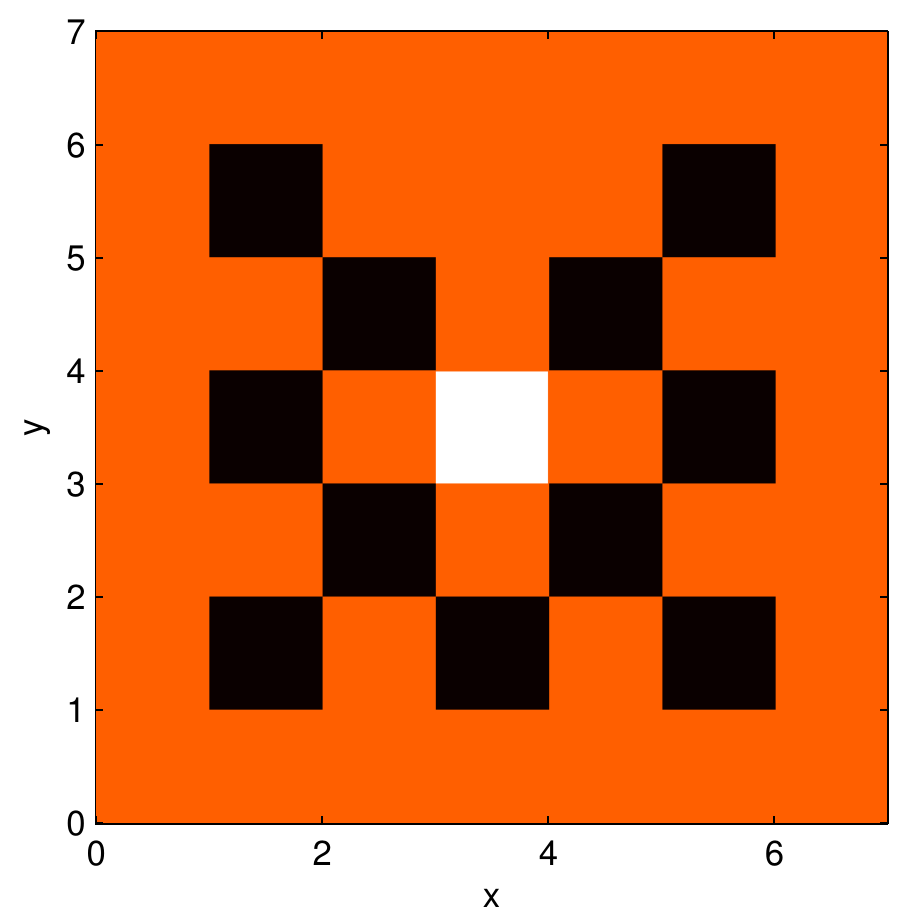}
\end{minipage}
\vspace{-.6em}
\caption{Checkerboard problem, material coefficients: isotropic source (white), purely scattering $\Sigma_{s0}=1=\Sigma_t$ (orange and white), purely absorbing $\Sigma_a=10=\Sigma_t$ (black).}
\label{fig:checkerboard_geometry}
\end{figure}

\begin{figure}
\begin{minipage}[b]{.49\textwidth}
\includegraphics[width=\textwidth]{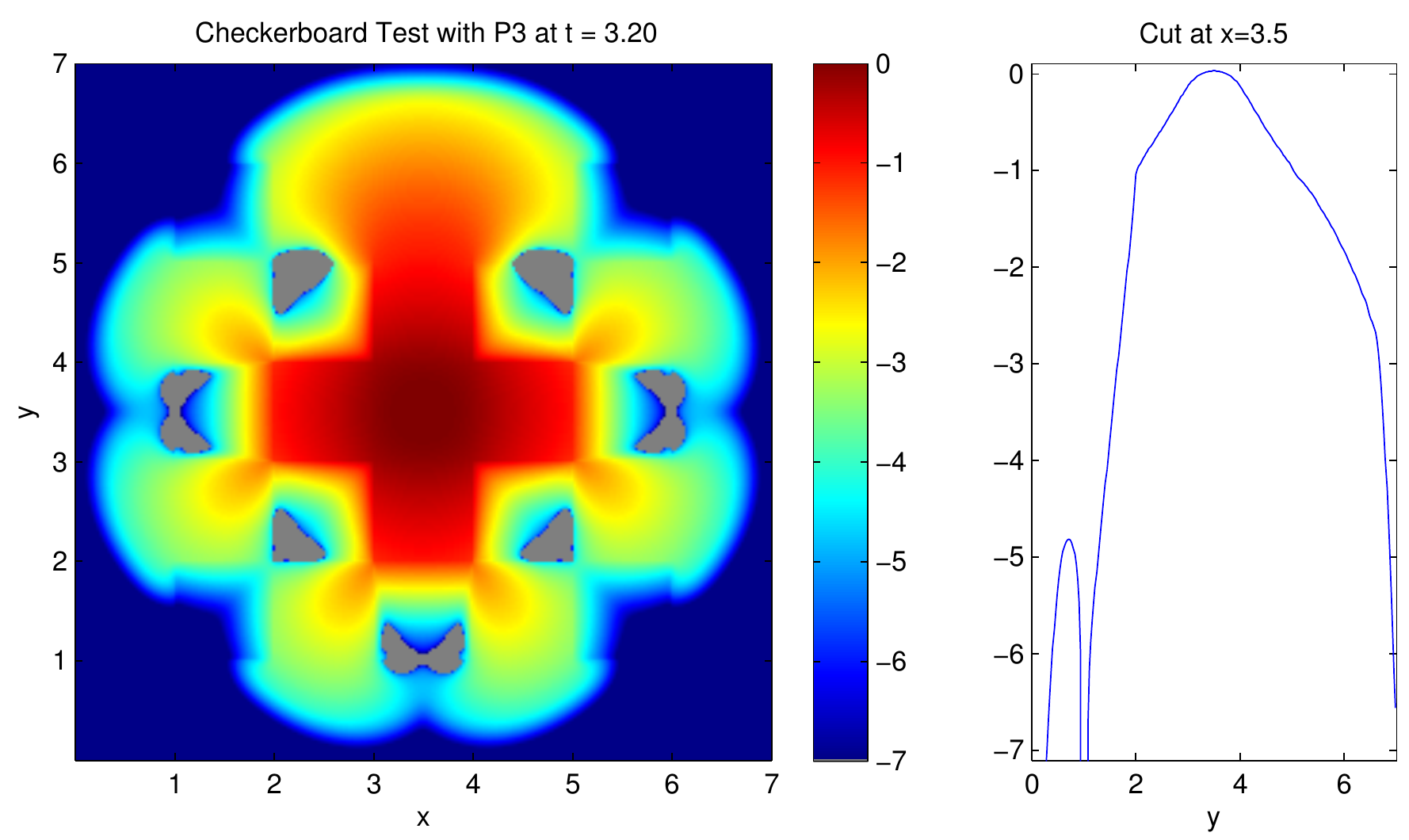}
\end{minipage}
\hfill
\begin{minipage}[b]{.49\textwidth}
\includegraphics[width=\textwidth]{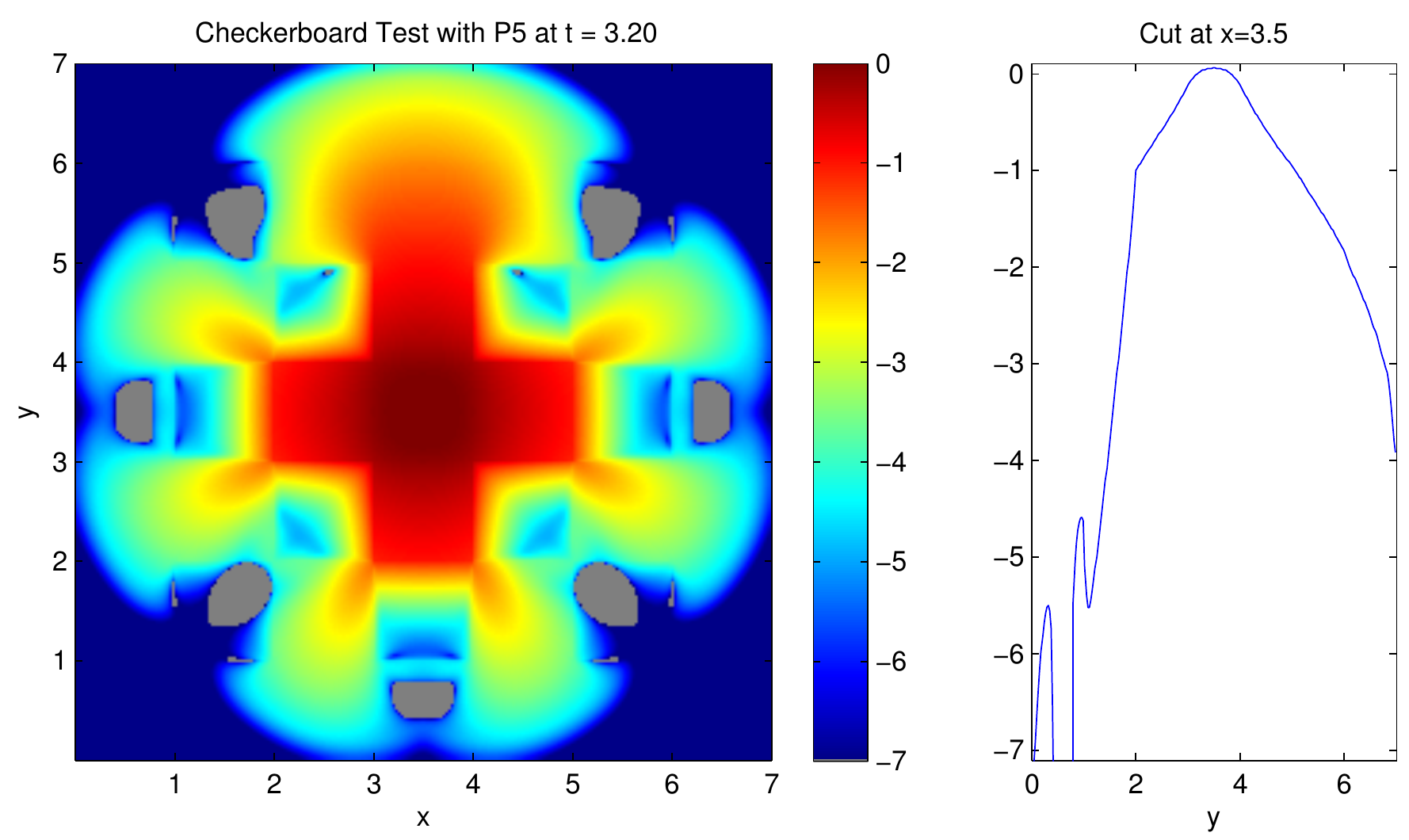}
\end{minipage}

\vspace{.5em}
\begin{minipage}[b]{.49\textwidth}
\includegraphics[width=\textwidth]{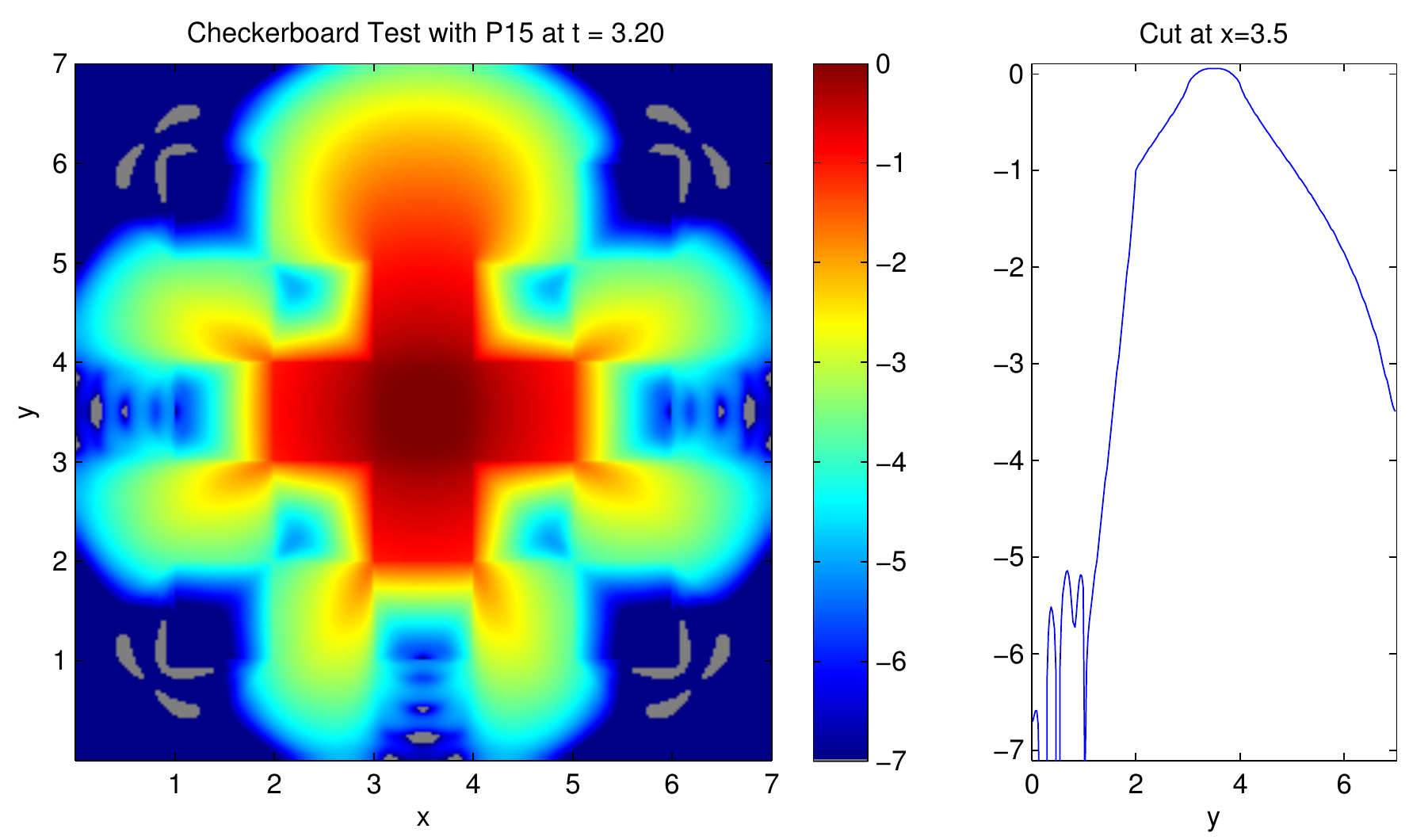}
\end{minipage}
\hfill
\begin{minipage}[b]{.49\textwidth}
\includegraphics[width=\textwidth]{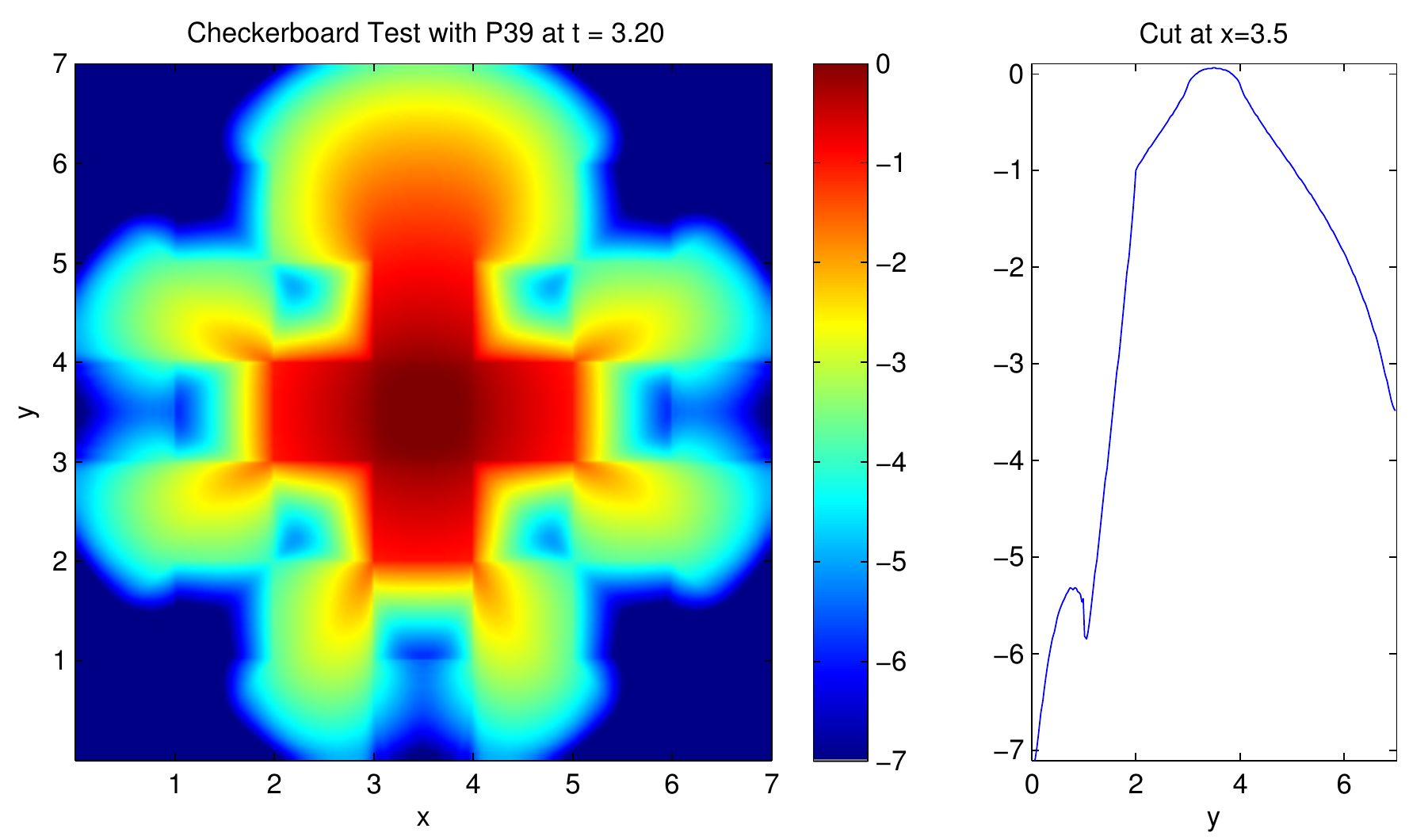}
\end{minipage}
\vspace{-.6em}
\caption{Checkerboard problem. Scalar flux $R_0^0$ at $t=3.2$ for several moment approximations ($P_3$, $P_5$, $P_{15}$, and $P_{39}$), computed with $250\times 250$ grid points. The values are plotted in a logarithmic scale and limited to seven orders of magnitude. The left panel shows the solution in a color plot, where negative values of $R_0^0$ are depicted in gray. The right panel shows the solution evaluated along the vertical line $x=3.5$.}
\label{fig:checkerboard}
\end{figure}

In the \textsf{StaRMAP} project, the checkerboard test case is implemented in the example file \texttt{starmap\_ex\_checkerboard.m}. Extrapolation boundary conditions\footnote{Most frequently, this test is equipped with vacuum boundary conditions, which are not implemented in the \textsf{StaRMAP} routines. The easiest way to ``emulate'' vacuum in \textsf{StaRMAP} for this test would be to extend the computational domain.} are enforced. At the initial time $t=0$, all quantities are zero. The spatial grid is $250\times 250$ for all cases.

We compare the scalar flux using different orders of $P_N$ approximations, with the scalar flux $R_0^0$ at $t=3.2$ shown in Fig.~\ref{fig:checkerboard}. The case $N=3$ corresponds to 10 angular degrees of freedom, while the case $N=39$ possesses 820 angular degrees of freedom. For increasing $N$, several well-known properties of the $P_N$ approximation can be observed. First, the maximum propagation speed of information approaches the correct limit of 1 which can be seen at the upper front. Second, the Gibbs phenomena in the solution (left, right, bottom) are diminished. The $P_{39}$ solution can be seen as a fully converged transport solution that in particular does not possess any negative scalar flux values.

To get an impression about the computational cost and efficiency of the \textsf{StaRMAP} solver: the computation of a $P_5$ solution on a $100\times 100$ spatial grid corresponds to 18.3 million degrees of freedom in space, angle, and time, and it takes about 1.5 seconds to compute on a 2008 Lenovo W500 laptop.

%---------------------------------------------------------------------------
\subsection{Line Source}
\label{subsec:line_source}
%---------------------------------------------------------------------------
The line source problem has been investigated for the $P_N$ equations in \cite{BrunnerHolloway2005}. Essentially one is trying to compute the Green's function for an initial isotropic Dirac mass at the origin, i.e.
\begin{equation*}
\psi(t=0,x,\Omega) = \frac{1}{4\pi}\delta(x)\;.
\end{equation*}
One can also imagine an infinite line emitting particles isotropically (hence the name). There exists a semi-analytical solution to the full transport equation \eqref{eq:RTE} for this problem due to Ganapol et al.~\cite{GanapolBakerDahlAlcouffe2001}. The exact solution consists of a circular front moving away from the origin as well as a tail of particles which have been scattered or not emitted perpendicularly from the line.

\begin{figure}
\begin{center}
\begin{minipage}[b]{.75\textwidth}
\includegraphics[width=\textwidth]{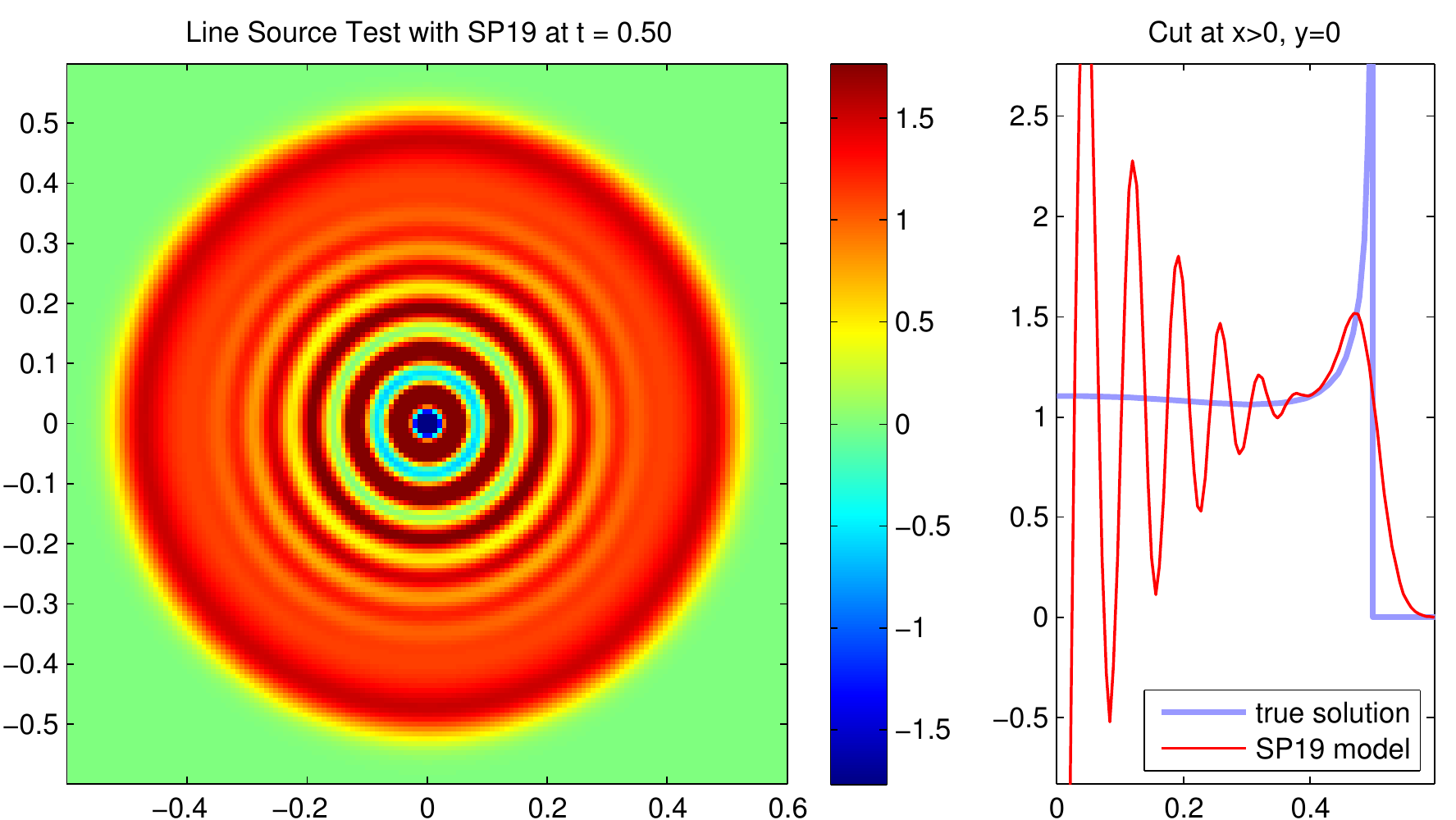}
\end{minipage}

\vspace{.5em}
\begin{minipage}[b]{.75\textwidth}
\includegraphics[width=\textwidth]{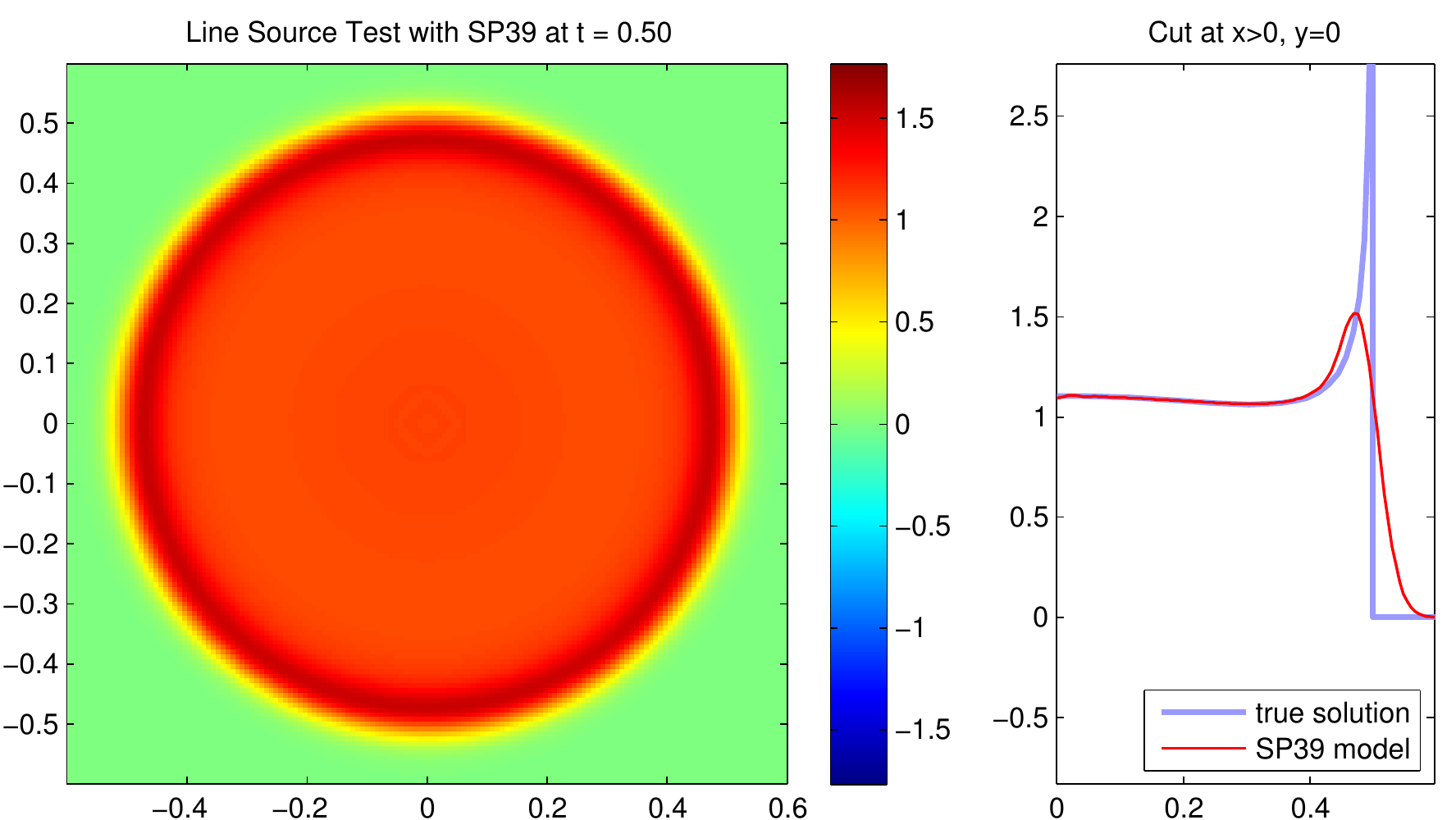}
\end{minipage}
\end{center}
\caption{Results of the line source problem at $t=0.5$, computed with $SP_{19}$ (top) and $SP_{39}$ (bottom). The left sub-figures show the 2d solution profile. The right sub-figures show the solution on a cut along the positive radial axis (red), together with Ganapol's semi-analytic solution (blue). Note that $SP_N$ and $P_N$ are equivalent in this case.}
\label{fig:linesource}
\end{figure}

When setting up this test case, some care is in order. Although Ganapol's semi-analytical benchmark solution does not consist of Diracs, it has a singularity at the edge. More precisely, the solution $R_0^0$ approaches infinity as $\frac{1}{r_0-r}$ (where $r$ is the distance from origin and $r_0$ is the front). On the other hand, Brunner \& Holloway \cite{BrunnerHolloway2005} have shown that the analytical solution of the $P_1$ equations for an initial Dirac consists of traveling Diracs at the front. This is also true for higher-order $P_N$ solutions. With a grid-based method as simple as ours, one cannot expect to capture either the Diracs or the singularities.

Instead, we realize the initial condition as a narrow Gaussian in space
\begin{equation}
\label{eq:Gaussian}
R_0^0 = \frac{1}{4\pi\sigma}\exp\prn{-\frac{x^2+y^2}{4\sigma}}
\end{equation}
with $\sigma = 3.2\times 10^{-4}$. This way, the integral of $R_0^0$ over the whole space is one, so it can be compared to the scalar flux of the benchmark solution.

We choose $\Sigma_a = 0$ and $\Sigma_{s} = \frac{1}{4\pi}$, so that $\Sigma_{s0} = 1$. The computational domain is $[-0.6,0.6]\times[-0.6,0.6]$.  In the \textsf{StaRMAP} project, the line source test case is implemented in the example file \texttt{starmap\_ex\_linesource.m}. This m-file also approximates Ganapol's semi-analytic solution with an accuracy that makes the resulting function indistinguishable from the exact solution when plotted.

At this point, two aspects should be stressed. A piece of good news is that the line source test case satisfies the conditions under which the $SP_N$ and the $P_N$ equations are equivalent (see \cite{McClarren2011}). We therefore can compute the numerical solution using the $SP_N$ equations, and thus obtain the same result as with the $P_N$ equations, but much faster. A note of caution must be given regarding the width of the initial Gaussian. The pseudo-time $\sigma$ must be chosen large enough so that the Gaussian is well resolved by the grid. If one fails to do so, spurious oscillations will arise (due to the absence of limiters in the numerical scheme) that can be quite detrimental to the quality of the numerical approximation.

The computational results for $SP_{19}$ and $SP_{39}$, both computed on a $150\times 150$ spatial grid, are shown in Figure \ref{fig:linesource}. These results are in line with the observations from \cite{BrunnerHolloway2005} and the previous test case. While in theory the $P_N$ equations have a rotationally symmetric solution, it has been observed in \cite{BrunnerHolloway2005} that a higher-order scheme is necessary to observe the rotational invariance numerically. In our case, both the $SP_{19}$ and the $SP_{39}$ solutions exhibit this behavior. For the lower-order $SP_{19}$ model, the Gibbs phenomenon is very clearly visible. The $SP_{39}$ solution, on the other hand, shows only very tiny oscillations. Of course, it cannot capture the singular behavior of the solution, but it does a good job at capturing the radiation front. These two examples are shown because the \textsf{StaRMAP} code solves them within a few minutes. However, if one is willing to accept longer computation times, then suitable grid refinement and a more narrowly focused initial condition would yield an even better approximation.

%---------------------------------------------------------------------------
\subsection{Beam in Void and Medium}
%---------------------------------------------------------------------------
In this problem we study the advection of a packet of particles that are emanating from a source (essentially compactly supported in space) and all move in the same direction in empty space, before they hit a material. This test investigates how well the method performs in a void, at material interfaces, and with anisotropic scattering.

The source can be written as
\begin{equation*}
q(t,x,\Omega) = \bar q(x) \delta(\Omega-\Omega_0)\;,
\end{equation*}
where $\bar q$ is a spatial distribution and $\Omega_0$ is a direction in the plane. Here we choose $\Omega_0$ to have an angle of $\pi/6$ with respect to the $x$-axis and $\bar q$ is the same spatial Gaussian as for the line source problem. The proper solution of this beam problem involves the challenge of calculating the correct moments of the source for all components of the $P_N$ solution vector. We do so by following the definitions and transformations given in Sect.~\ref{sec:PN_2d}.

The domain for this test is $[-0.6,0.6]\times [-0.6,0.6]$ with extrapolation boundary conditions on all sides. The region with $x<0.3$ contains no material, i.e.\ $\Sigma_t = 0$, whereas the region $x\ge 0.3$ is a material with no absorption ($\Sigma_a = 0$) but anisotropic Henyey-Greenstein \cite{HenyeyGreenstein1941} scattering $\Sigma_{s\ell} = 100 g^\ell$ with $g = 0.85$.

In the \textsf{StaRMAP} project, this test case is realized similarly to the manufactured solution example, namely the file \texttt{starmap\_create\_beam.m} performs the computation of the beam initial conditions, and then generates an example file \texttt{starmap\_ex\_beam\_auto.m} that can then be executed to produce the plots shown in Fig.~\ref{fig:beam}.

\begin{figure}
\begin{minipage}[b]{.485\textwidth}
\includegraphics[width=\textwidth]{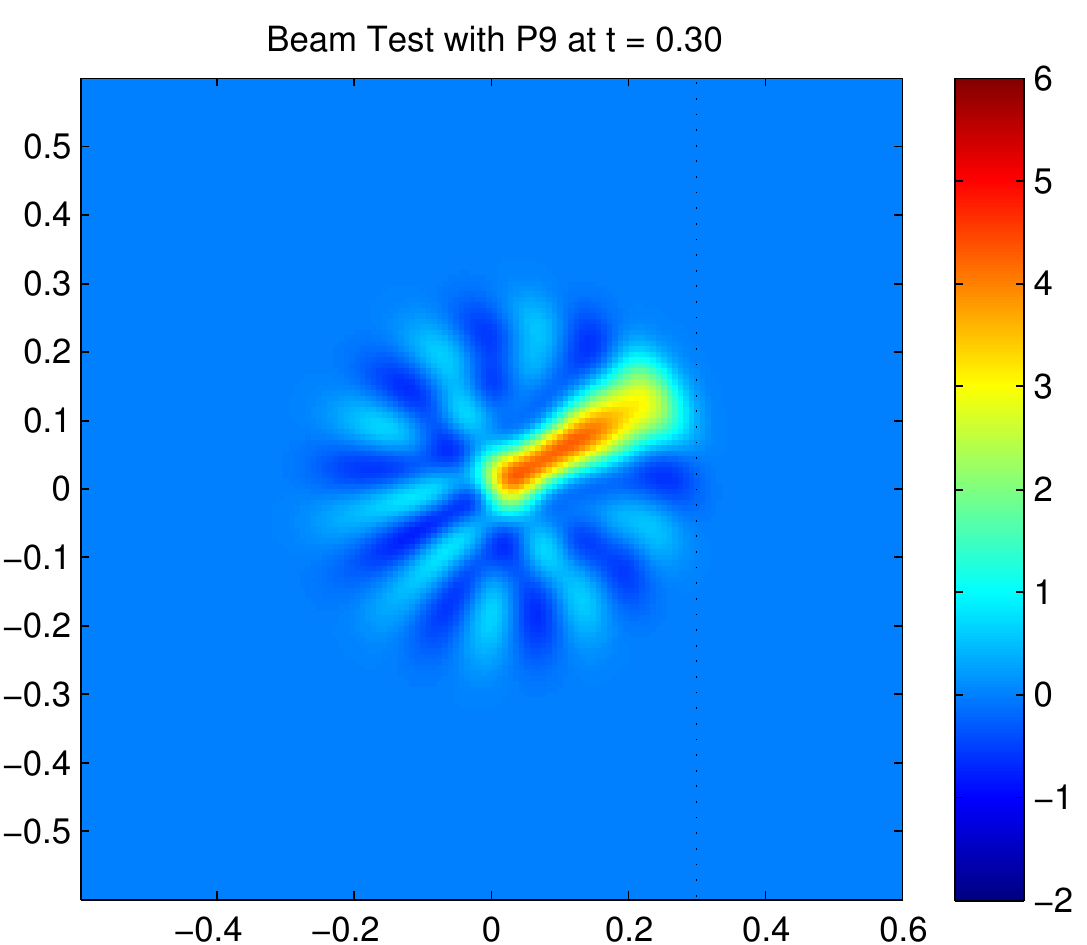}
\end{minipage}
\hfill
\begin{minipage}[b]{.485\textwidth}
\includegraphics[width=\textwidth]{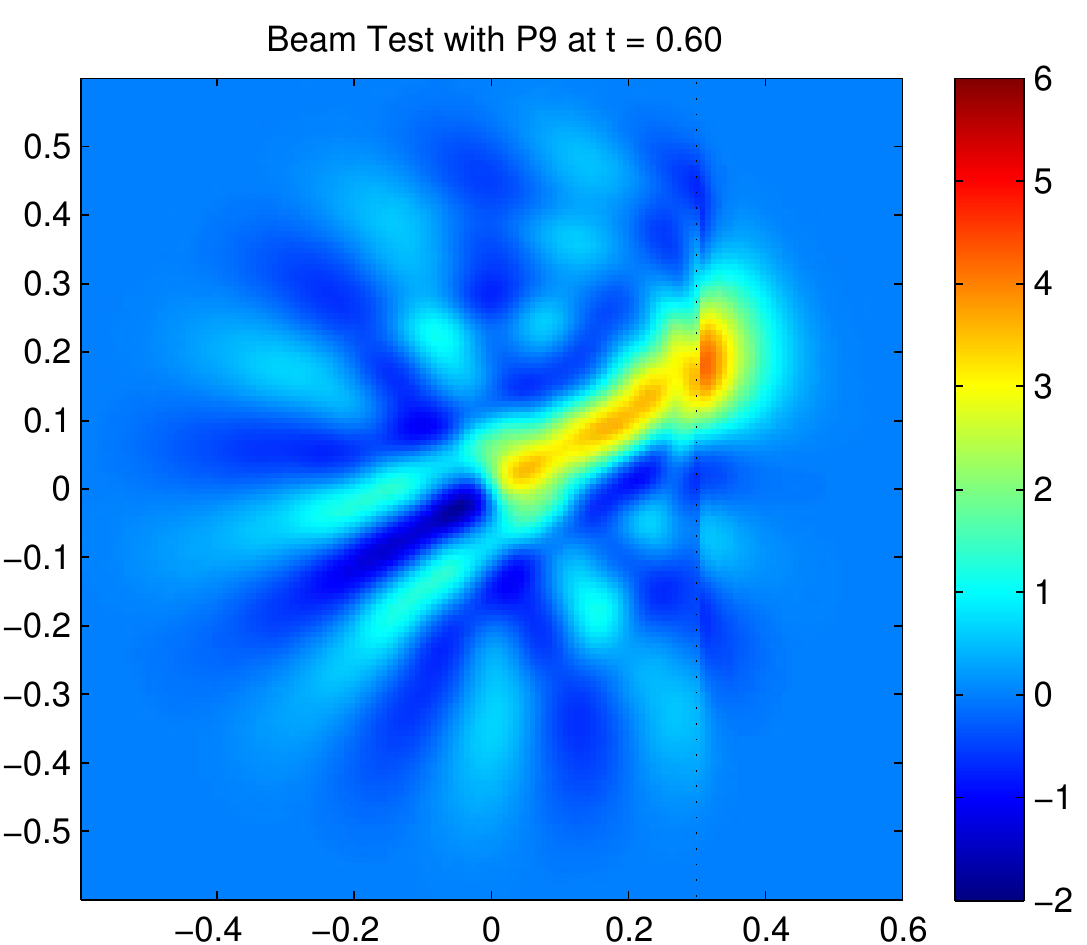}
\end{minipage}

\vspace{.5em}
\begin{minipage}[b]{.485\textwidth}
\includegraphics[width=\textwidth]{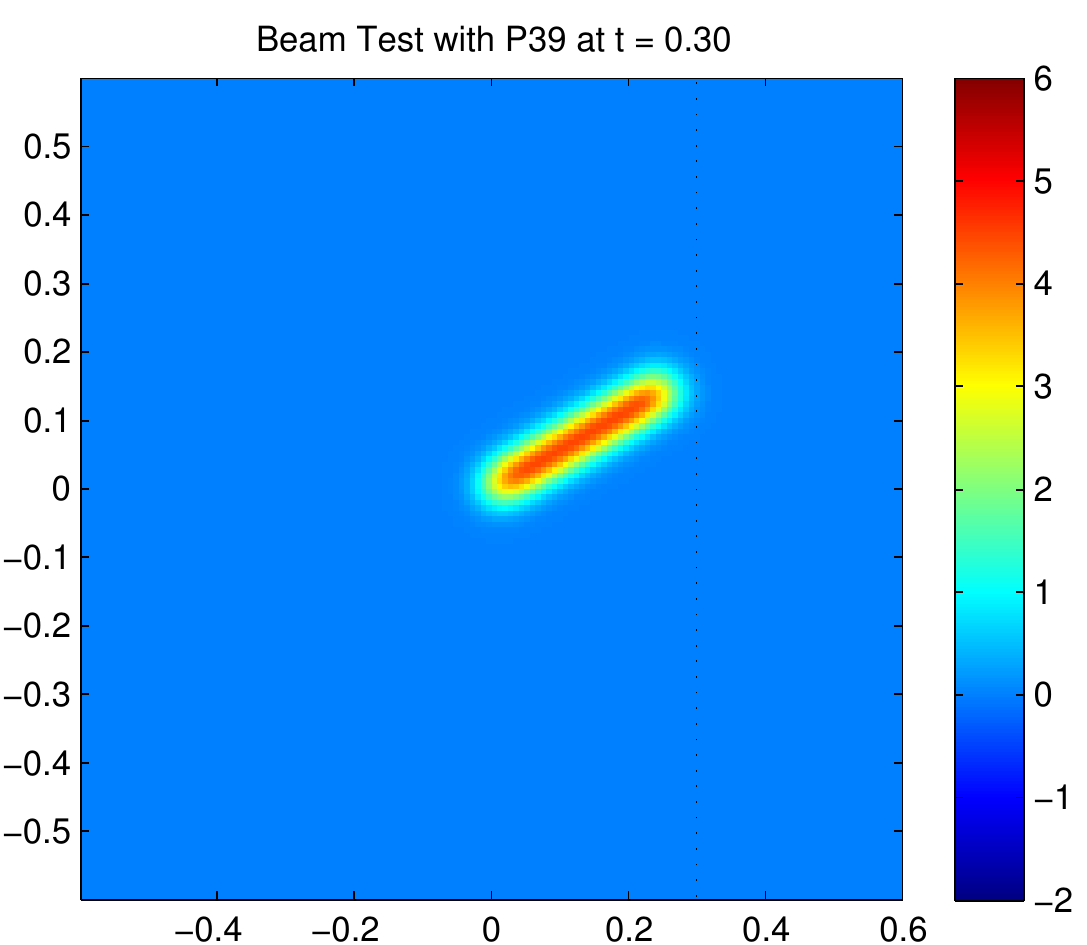}
\end{minipage}
\hfill
\begin{minipage}[b]{.485\textwidth}
\includegraphics[width=\textwidth]{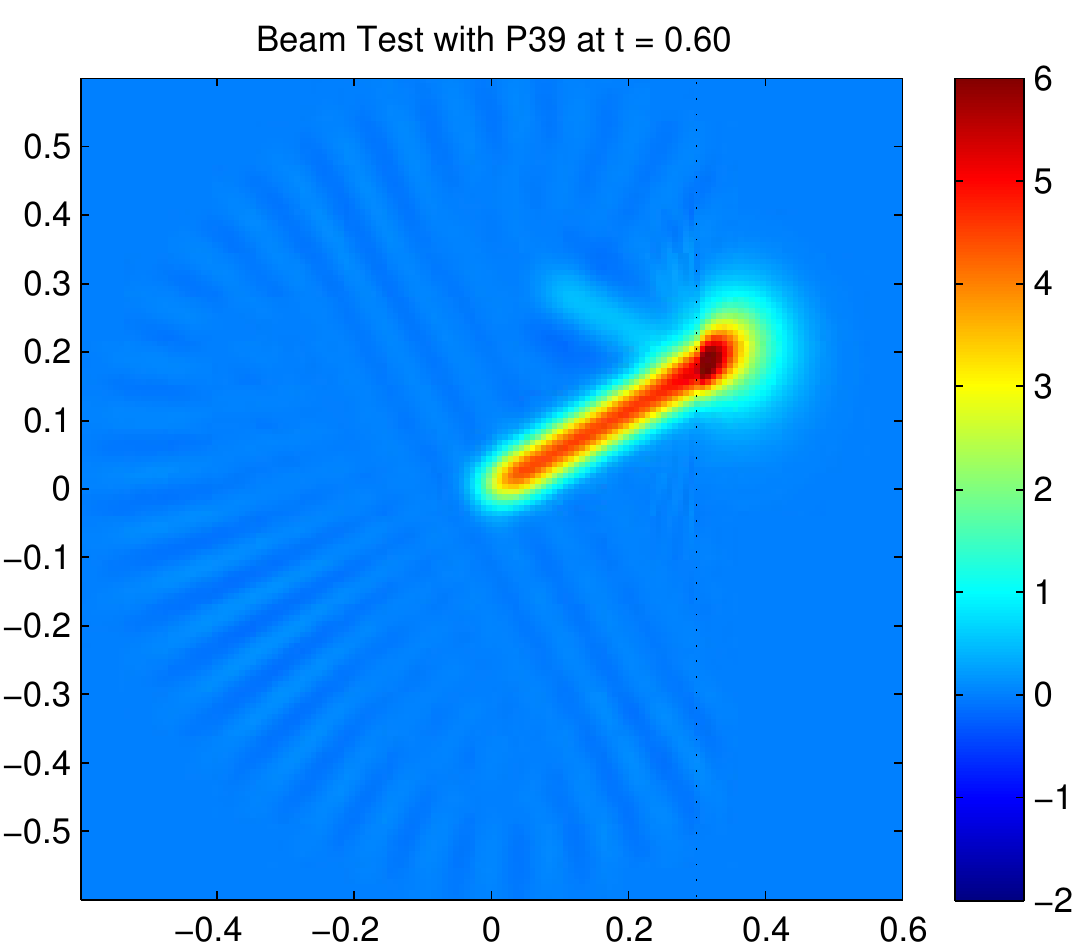}
\end{minipage}
\caption{Beam problem, computed with $P_9$ (top row) and $P_{39}$ (bottom row). Shown is the scalar flux $R_0^0$ at times $t=0.3$ (left column) and $t=0.6$ (right column). A spatially Gaussian source propagates into the direction with angle $\pi/6$ with respect to $x$-axis, and hits a medium at $x=0.3$.}
\label{fig:beam}
\end{figure}

Figure~\ref{fig:beam} shows the scalar flux $R_0^0$ at two times ($t\in\{0.3,0.6\}$) for a beam computation on a spatial grid of $150\times 150$ cells, using two moment models: $P_9$ (top two figures) and $P_{39}$ (bottom two figures). The full transport equation \eqref{eq:RTE} would initially (i.e., before the beam hits the medium) just advect the source with unit speed into direction $\Omega_0$. Due to the truncation in the $P_N$ approximation, the beam smears out, and Gibbs artifacts appear during the evolution. These are clearly visible in the $P_9$ solution (top left figure), in which the beam spreads out and oscillations occur. In contrast, the $P_{39}$ solution at $t=0.3$ (bottom left figure) is very close to a fully converged transport solution; no truncation artifacts are visible.

At the material interface, particles start to scatter strongly, albeit in a forward direction. As a consequence, the propagation of the beam effectively slows down and it is smeared out. The $P_{39}$ solution at $t=0.6$ (bottom right figure) shows the correct beam profile (including a reflected beam at the interface between medium and void), albeit a small amount of Gibbs artifacts that are visible. The corresponding $P_9$ solution (top right figure) is of significantly lower accuracy: while the general direction of the radiative flux is represented correctly, the precise shape of the radiation profile (e.g., beam thickness, reflections, maxima) is not well-captured.

When considering this test, it must be remarked that beams are not the type of phenomena that one would usually consider using moment methods for. In addition, moment approximations to problems with voids in general do not converge to the correct steady state solution as $t\to\infty$ (the Gibbs oscillations would amplify in time). Nevertheless, as the results in Fig.~\ref{fig:beam} show, and as one can verify with the \textsf{StaRMAP} project files, beams can be reasonably well approximated. Thus, with the moment order $N$ chosen suitably large and for sufficiently short times, challenging problems, such as beams moving from a void into a highly scattering medium, can be computed accurately with \textsf{StaRMAP}.

%---------------------------------------------------------------------------
\subsection{Further Test Cases}
%---------------------------------------------------------------------------
The \textsf{StaRMAP} code has also been applied to several other examples. Among them are:
\begin{itemize}
\item The ``boxes'' test, implemented in the example file \texttt{starmap\_ex\_boxes.m}, is a test case that was proposed by McClarren \cite{McClarren2011} to demonstrate the equivalence of the $SP_N$ and the $P_N$ equations under certain assumptions.
\item The ``control rod'' test was devised by Olbrant \cite{OlbrantLarsenFrankSeibold2013} as an example for time-dependent material coefficients. The setup can be interpreted as a simple model for the evolution of the neutron density in a nuclear reactor, when an absorbing rod is moved into and out of the domain. In the \textsf{StaRMAP} project, it is implemented in the example file \texttt{starmap\_ex\_controlrod.m}.
\end{itemize}
Further details and numerical results for both tests are given in \cite{OlbrantLarsenFrankSeibold2013}.

%===========================================================================
\vspace{1.5em}
\section{Conclusions and Outlook}
\label{sec:conclusions_and_outlook}
%===========================================================================
The presented numerical approach to solve the spherical harmonics $P_N$ equations (and variants thereof) of radiative transfer has been demonstrated to be highly flexible, and to be applicable to a wide variety of radiative transfer problems in two space dimensions. The method is implemented in the \textsc{Matlab} project \textsf{StaRMAP}, that can be downloaded \cite{StaRMAP}, and all examples presented in this paper can be reproduced by the user.

As shown in this paper, the staggered grid finite difference approach employed in \textsf{StaRMAP} is second order accurate, and stable independent of the magnitude of the material scattering and absorption coefficients (with the CFL condition that is commonly incurred for advection problems). Further advantages and drawbacks of the method have been discussed in detail.

The availability of the complete \textsc{Matlab} code serves another purpose besides the reproducibility of the results presented in this paper: we encourage users to create their own \textsf{StaRMAP} example files for test cases that arise in their own work. The modular structure of the code facilitates this goal: in general, users should not be required to modify the solver file \texttt{starmap\_solver.m}, but instead they can simply build on the existing example files.

While the current \textsf{StaRMAP} code is quite general in that it allows for arbitrary moment order and for various types of moment systems, it has a number of crucial limitations. Some cannot be overcome easily, such as the occurrence of Gibbs phenomena or the regularity of the grid structures (see Sect.~\ref{subsec:limitations} for more details). In turn, some limitations can (and shall) be addressed in the future, as follows.

In principle, the presented methodology applies to the full 3D case as well. We have decided against generalizing the current \textsf{StaRMAP} code to 3D, mainly for the simple reason of computational cost and memory requirements. While 2D problems can be solved in \textsc{Matlab} with a good accuracy, in reasonable compute times, and with a reasonable amount of memory consumption, we do not perceive things as quite mature enough yet for a 3D \textsc{Matlab} version of \textsf{StaRMAP}.

The \textsf{StaRMAP} code can be extended to other moment models that exhibit a similar coupling structure between the moments as the $P_N$ and the $SP_N$ equations. Specifically, the filtered $P_N$ ($FP_N$) method \cite{McClarrenHauck2010} falls into this category, as do the diffusion-corrected $P_N$ ($D_N$) equations, which contain an additional diffusion term in the highest order moments. These last types of models can be rationalized via asymptotic analysis \cite{SchaeferFrankLevermore2011} or by the Mori-Zwanzig formalism of statistical mechanics \cite{FrankSeibold2011}. The splitting procedure employed in the \textsf{StaRMAP} code allows for the addition of a diffusion step applied to the moments of the highest order.
Finally, it must be verified that the method presented here satisfies an important property for numerical approaches for radiative transfer, namely whether it is asymptotic-preserving (AP) \cite{Jin1999}. It can be quite easily shown that the scheme proposed in this paper is at least formally AP. Moreover, as has been suggested in \cite{Jin2012}, because of the staggered grid we obtain a compact stencil in the diffusive limit. In addition, due to the explicit integration in time, the time step does not depend on the possibly stiff right hand side. The AP property will be investigated in detail in a companion paper \cite{FrankSeibold2014}.

%===========================================================================
\vspace{1.5em}
\section*{Acknowledgments}
%===========================================================================
The authors would like to thank E. Olbrant for helpful discussions. B. Seibold would like to acknowledge the support by the National Science Foundation. B. Seibold was supported through grant DMS--1115269. In addition, B. Seibold wishes to acknowledge partial support by the NSF through grants DMS--1007967 and DMS--1318709. M. Frank has been supported by the German Research Foundation DFG under contract number FR 2841/1-1.

%===========================================================================
\bibliographystyle{plain}
\bibliography{references_complete}
%===========================================================================

\end{document}